\def\final{1}

\documentclass[11pt]{article}

\usepackage{algorithm, algorithmic}
\usepackage{amsmath, amssymb, amsthm}
\usepackage{enumerate}
\usepackage{enumitem}
\usepackage{framed}
\usepackage[margin=1.1in]{geometry}
\usepackage{verbatim}
\usepackage{color}
\usepackage{microtype}
\usepackage{kpfonts}
	\DeclareMathAlphabet{\mathsf}{OT1}{cmss}{m}{n}
	\SetMathAlphabet{\mathsf}{bold}{OT1}{cmss}{bx}{n}

\usepackage{multirow, tabularx}
\definecolor{DarkGreen}{rgb}{0.15,0.5,0.15}
\definecolor{DarkRed}{rgb}{0.6,0.2,0.2}
\definecolor{DarkBlue}{rgb}{0.15,0.15,0.55}
\definecolor{DarkPurple}{rgb}{0.4,0.2,0.4}

\usepackage[pdftex]{hyperref}
\hypersetup{
    linktocpage=true,
    colorlinks=true,				
    linkcolor=DarkBlue,		
    citecolor=DarkBlue,	
    urlcolor=DarkBlue,		
}

\ifnum\final=0
\newcommand{\mynote}[2]{{\color{#1} \marginpar{\tiny #2}}}
\newcommand{\mybignote}[2]{{\color{#1} $\langle \langle$ #2$\rangle \rangle$}}
\else
\newcommand{\mynote}[2]{}
\newcommand{\mybignote}[2]{}
\fi
\newcommand{\jnote}[1]{\mynote{DarkRed}{Jon: {#1}}}
\newcommand{\bigjnote}[1]{\mybignote{DarkRed}{Jon: #1}}
\newcommand{\mnote}[1]{\mynote{DarkBlue}{Mark: {#1}}}

\newcommand{\tnote}[1]{\mynote{DarkGreen}{Thomas: {#1}}}

\newcommand{\INDSTATE}[1][1]{\STATE\hspace{#1\algorithmicindent}}

\newcolumntype{Y}{>{\centering\arraybackslash}X}

\newcommand{\pr}[2]{\underset{#1}{\mathbb{P}}\left[ #2 \right]}

\newcommand{\ex}[2]{\underset{#1}{\mathbb{E}}\left[ #2 \right]}
\newcommand{\var}[2]{\underset{#1}{\mathrm{Var}}\left[ #2 \right]}

\newcommand{\poly}{\mathrm{poly}}

\newcommand{\zo}{\{0,1\}}

\newcommand{\pmo}{\{\pm1\}}
\newcommand{\getsr}{\gets_{\mbox{\tiny R}}}

\newcommand{\set}[1]{\left\{#1\right\}}
\newcommand{\from}{:}

\newcommand{\loss}{L}
\newcommand{\eps}{\varepsilon}

\DeclareMathOperator*{\argmax}{arg\,max}

\newcommand{\offlinealg}[2]{\mathsf{Offline}_{#1 {\rightarrow \atop \leftarrow} #2}}
\newcommand{\onlinealg}[2]{\mathsf{Online}_{#1 {\rightarrow \atop \leftarrow} #2}} 
\newcommand{\adaptivealg}[2]{\mathsf{Adaptive}_{#1 {\rightarrow \atop \leftarrow} #2}}

\newcommand{\N}{\mathbb{N}}
\newcommand{\R}{\mathbb{R}}

\newcommand{\cM}{\mathcal{M}}

\newcommand{\cQ}{\mathcal{Q}}
\newcommand{\cR}{\mathcal{R}}

\newtheorem{theorem}{Theorem}[section]
\newtheorem{thm}[theorem]{Theorem}
\newtheorem{lemma}[theorem]{Lemma}
\newtheorem{lem}[theorem]{Lemma}
\newtheorem{fact}[theorem]{Fact}

\newtheorem{claim}[theorem]{Claim}
\newtheorem{clm}[theorem]{Claim}
\newtheorem{remark}[theorem]{Remark}
\newtheorem{corollary}[theorem]{Corollary}
\newtheorem{prop}[theorem]{Proposition}

\theoremstyle{definition}
\newtheorem{defn}[theorem]{Definition}
\newtheorem{definition}[theorem]{Definition}

\newcommand{\Lap}{\operatorname{Lap}}

\newcommand{\thresh}[1]{\mathsf{Thresh}_{#1}}

\newcommand{\nope}[1]{}
\newcommand{\univ}{X}
\newcommand{\db}{x}
\newcommand{\row}{x}
\newcommand{\alg}{M}
\newcommand{\adv}{A}
\newcommand{\range}{\cR}

\newcommand{\qprefix}{Q_{\textrm{prefix}}}
\newcommand{\qcorr}{Q_{\textrm{corr}}}
\newcommand{\qthresh}{Q_{\textrm{thresh}}}

\newcommand{\lT}{t_\ell}
\newcommand{\uT}{t_u}
\newcommand{\hlT}{\hat{t}_\ell}
\newcommand{\huT}{\hat{t}_u}
\newcommand{\leftsymb}{\operatorname{L}}
\newcommand{\rightsymb}{\operatorname{R}}
\newcommand{\haltsymb}{\top}

\title{Make Up Your Mind: \\ The Price of Online Queries in Differential Privacy}
\author{Mark Bun\thanks{Harvard University John A.~Paulson School of Engineering and Applied Sciences. Supported by an NDSEG Fellowship and NSF grant CNS-1237235. Part of this work was done while the author was visiting Yale University.  \texttt{mbun@seas.harvard.edu}} \and Thomas Steinke\thanks{Harvard University John A. Paulson School of Engineering and Applied Sciences.  Supported by NSF grants CCF-1116616, CCF-1420938, and CNS-1237235. \texttt{tsteinke@seas.harvard.edu} } \and Jonathan Ullman\thanks{Northeastern University College of Computer and Information Science.  \texttt{jullman@ccs.neu.edu}}}

\begin{document}

\maketitle

\pagenumbering{gobble}
\begin{abstract}
We consider the problem of answering queries about a sensitive dataset subject to differential privacy.  The queries may be chosen adversarially from a larger set $Q$ of allowable queries in one of three ways, which we list in order from easiest to hardest to answer:
\begin{itemize}
\item {\em Offline:} The queries are chosen all at once and the differentially private mechanism answers the queries in a single batch.
\item {\em Online:} The queries are chosen all at once, but the mechanism only receives the queries in a streaming fashion and must answer each query before seeing the next query.
\item {\em Adaptive:} The queries are chosen one at a time and the mechanism must answer each query before the next query is chosen. In particular, each query may depend on the answers given to previous queries.
\end{itemize}
Many differentially private mechanisms are just as efficient in the adaptive model as they are in the offline model. Meanwhile, most lower bounds for differential privacy hold in the offline setting. This suggests that the three models may be equivalent.

We prove that these models are all, in fact, distinct.  Specifically, we show that there is a family of statistical queries such that exponentially more queries from this family can be answered in the offline model than in the online model.  We also exhibit a family of search queries such that exponentially more queries from this family can be answered in the online model than in the adaptive model.  We also investigate whether such separations might hold for simple queries like threshold queries over the real line.
\end{abstract}

\vfill
\newpage

\tableofcontents

\vfill
\newpage

\pagenumbering{arabic}
\section{Introduction}
\emph{Differential privacy}~\cite{DworkMNS06} is a formal guarantee that an algorithm run on a sensitive dataset does not reveal too much about any individual in that dataset.  Since its introduction, a rich literature has developed to determine what statistics can be computed accurately subject to differential privacy.  For example, suppose we wish to approximate a real-valued \emph{query} $q(\db)$ on some dataset $\db$ that consists of the private data of many individuals.  Then, this question has a clean answer---we can compute a differentially private estimate of $q(\db)$ with error proportional to the \emph{global sensitivity} of $q,$ and we cannot have smaller error in the worst case.

But how much error do we need to answer a large set of queries $q_1,\dots,q_k$?  Before we can answer this question, we have to define a model of how the queries are asked and answered.  The literature on differential privacy has considered three different interactive models\footnote{Usually, the ``interactive model'' refers only to what we call the ``adaptive model.''  We prefer to call all of these models interactive, since they each require an interaction with a data analyst who issues the queries.  We use the term ``interactive'' to distinguish these models from one where the algorithm only answers a fixed set of queries.} for specifying the queries:
\begin{itemize}
\item {\em The Offline Model:} The sequence of queries $q_1,\dots,q_k$ are given to the algorithm together in a batch and the mechanism answers them together.
\item {\em The Online Model:} The sequence of queries $q_1,\dots,q_k$ is chosen in advance and then the mechanism must answer each query $q_j$ before seeing $q_{j+1}$.
\item {\em The Adaptive Model:} The queries are not fixed in advance, each query $q_{j+1}$ may depend on the answers to queries $q_1,\dots,q_j$.
\end{itemize}
In all three cases, we assume that $q_1, \cdots, q_k$ are chosen from some family of allowable queries $Q$, but may be chosen adversarially from this family.

Differential privacy seems well-suited to the adaptive model.  Arguably its signature property is that any adaptively-chosen sequence of differentially private algorithms remains collectively differentially private, with a graceful degradation of the privacy parameters~\cite{DworkMNS06, DworkRV10}.  As a consequence, there is a simple differentially private algorithm that takes a dataset of $n$ individuals and answers $\tilde{\Omega}(n)$ statistical queries in the adaptive model with error $o(1/\sqrt{n})$, simply by perturbing each answer independently with carefully calibrated noise.  In contrast, the seminal lower bound of Dinur and Nissim and its later refinements~\cite{DinurN03, DworkY08} shows that there exists a fixed set of $O(n)$ queries that cannot be answered by any differentially private algorithm with such little error, even in the easiest offline model.  For an even more surprising example, the private multiplicative weights algorithm of Hardt and Rothblum~\cite{HardtR10} can in many cases answer an exponential number of arbitrary, adaptively-chosen statistical queries with a strong accuracy guarantee, whereas~\cite{BunUV14} show that the accuracy guarantee of private multiplicative weights is nearly optimal even for a simple, fixed family of queries.

These examples might give the impression that answering adaptively-chosen queries comes ``for free'' in differential privacy---that everything that can be achieved in the offline model can be matched in the adaptive model.  Beyond just the lack of any separation between the models, many of the most powerful differentially private algorithms in all of these models use techniques from no-regret learning, which are explicitly designed for adaptive models.

Another motivation for studying the relationship between these models is the recent line of work connecting differential privacy to statistical validity for \emph{adaptive data analysis}~\cite{HardtU14, DworkFHPRR15, SteinkeU15a, BassilyNSSSU16}, which shows that differentially private algorithms for adaptively-chosen queries in fact yield state-of-the-art algorithms for statistical problems unrelated to privacy.  This connection further motivates studying the adaptive model and its relationship to the other models in differential privacy.

In this work, we show for the first time that these three models are actually distinct.  In fact, we show exponential separations between each of the three models.  These are the first separations between these models in differential privacy.

\subsection{Our Results}

Given a dataset $\db$ whose elements come from a data universe $\univ$, a \emph{statistical query on $\univ$} is defined by a predicate $\phi$ on $\univ$ and asks ``what fraction of elements in the dataset satisfy $\phi$?''  The answer to a statistical query lies in $[0,1]$ and our goal is to answer these queries up to some small additive error $\pm \alpha$, for a suitable choice of $0 < \alpha < 1$.  If the mechanism is required to answer \emph{arbitrary} statistical queries, then the offline, online, and adaptive models are essentially equivalent --- the upper bounds in the adaptive model match the lower bounds in the offline model \cite{DworkRV10,HardtR10,BunUV14,SteinkeU15b}.  However, we show that when the predicate $\phi$ is required to take a specific form, then it becomes strictly easier to answer a set of these queries in the offline model than it is to answer a sequence of queries presented online.
\begin{theorem}[Informal] \label{thm:mainprefix}
There exists a data universe $\univ$ and a family of statistical queries $Q$ on $\univ$ such that for every $n \in \N$,
\begin{enumerate}
\item there is a differentially private algorithm that takes a dataset $\db \in \univ^n$ and answers any set of $k = 2^{\Omega(\sqrt{n})}$ offline queries from $Q$ up to error $\pm 1/100$ from $Q$, but
\item no differentially private algorithm can take a dataset $\db \in \univ^n$ and answer an arbitrary sequence of $k = O(n^2)$ online (but not adaptively-chosen) queries from $Q$ up to error $\pm 1/100$.
\end{enumerate}
\end{theorem}

This result establishes that the online model is strictly harder than the offline model.  We also demonstrate that the adaptive model is strictly harder than the online model.  Here, the family of queries we use in our separation is not a family of statistical queries, but is rather a family of \emph{search queries} with a specific definition of accuracy that we will define later.
\begin{theorem}[Informal] \label{thm:maincorr}
For every $n \in \N$, there is a family of ``search'' queries $Q$ on datasets in $\univ^n$ such that
\begin{enumerate}
\item there is a differentially private algorithm that takes a dataset $\db \in \pmo^n$ and accurately answers any online (but not adaptively-chosen) sequence of $k = 2^{\Omega(n)}$ queries from $Q$, but
\item no differentially private algorithm can take a dataset $\db \in \pmo^n$ and accurately answer an adaptively-chosen sequence of $k = O(1)$ queries from $Q$. 
\end{enumerate}
\end{theorem}
We leave it as an interesting open question to separate the online and adaptive models for statistical queries, or to show that the models are equivalent for statistical queries.

Although Theorems~\ref{thm:mainprefix} and~\ref{thm:maincorr} separate the three models, these results use somewhat contrived families of queries.  Thus, we also investigate whether the models are distinct for \emph{natural} families of queries that are of use in practical applications.  One very well studied class of queries is \emph{threshold queries}.  These are a family of statistical queries $\qthresh$ defined on the universe $[0,1]$ and each query is specified by a point $\tau \in [0,1]$ and asks ``what fraction of the elements of the dataset are at most $\tau$?''  If we restrict our attention to so-called pure differential privacy (i.e.~$(\eps, \delta)$-differential privacy with $\delta = 0$), then we obtain an exponential separation between the offline and online models for answering threshold queries.

\begin{theorem}[Informal] \label{thm:thresholdsep}
For every $n \in \N$,
\begin{enumerate}
\item there is a pure differentially private algorithm that takes a dataset $\db \in [0,1]^n$ and answers any set of $k = 2^{\Omega(n)}$ offline queries from $\qthresh$ up to error $\pm 1/100$, but
\item no pure differentially private algorithm takes a dataset $\db \in [0,1]^n$ and answers an arbitrary sequence of $k = O(n)$ online (but not adaptively-chosen) queries from $\qthresh$ up to error $\pm 1/100$.
\end{enumerate}
\end{theorem}
We also ask whether or not such a separation exists for arbitrary differentially private algorithms (i.e.~$(\eps, \delta)$-differential privacy with $\delta > 0$). Theorem \ref{thm:thresholdsep} shows that, for pure differential privacy, threshold queries have near-maximal sample complexity.  That is, up to constants, the lower bound for online threshold queries matches what is achieved by the Laplace mechanism, which is applicable to arbitrary statistical queries. This may lead one to conjecture that adaptive threshold queries also require near-maximal sample complexity subject to approximate differential privacy. However, we show that this is not the case:\tnote{Rewrote this para, please read.}
\begin{theorem}
For every $n \in \N$, there is a differentially private algorithm that takes a dataset $\db \in [0,1]^n$ and answers any set of $k = 2^{{\Omega}({n})}$ adaptively-chosen queries from $\qthresh$ up to error $\pm 1/100$.
\end{theorem}
In contrast, for any offline set of $k$ thresholds $\tau_1,\dots,\tau_k$, we can round each element of the dataset up to an element in the finite universe $\univ = \set{\tau_1,\dots,\tau_k, 1}$ without changing the answers to any of the queries.  Then we can use known algorithms for answering all threshold queries over any finite, totally ordered domain~\cite{BeimelNS13, BunNSV15} to answer the queries using a very small dataset of size $n = 2^{O(\log^*(k))}$.  We leave it as an interesting open question to settle the complexity of answering adaptively-chosen threshold queries in the adaptive model.

\subsection{Techniques}
\subsubsection*{Separating Offline and Online Queries}
To prove Theorem~\ref{thm:mainprefix}, we construct a sequence of queries $q_1, \cdots, q_k$ such that, for all $j \in [k]$,
\begin{itemize}
\item $q_j$ ``reveals'' the answers to $q_1, \cdots, q_{j-1}$, but
\item $q_1, \cdots, q_{j-1}$ do not reveal the answer to $q_j$.
\end{itemize}
Thus, given the sequence $q_1, \cdots, q_k$ in the offline setting, the answers to $q_1, \cdots, q_{k-1}$ are revealed by $q_k$. So only $q_k$ needs to be answered and the remaining query answers can be inferred. However, in the online setting, each query $q_{j-1}$ must be answered before $q_j$ is presented and this approach does not work. This is the intuition for our separation.

To prove the online lower bound, we build on a lower bound for marginal queries~\cite{BunUV14}, which is based on the existence of short secure fingerprinting codes~\cite{BonehS98, Tardos03}.  Consider the data universe $\pmo^{k}$.  Given a dataset $\db \in \pmo^{n \times k}$, a \emph{marginal query} is a specific type of statistical query that asks for the mean of a given column of $\db$.  Bun et al.~\cite{BunUV14} showed that unless $k \ll n^2$, there is no differentially private algorithm that answers all $k$ marginal queries with non-trivial accuracy. This was done by showing that such an algorithm would violate the security of a short fingerprinting code due to Tardos~\cite{Tardos03}. We are able to ``embed'' $k$ marginal queries into the sequence of online queries $q_1, \cdots, q_k$. Thus a modification of the lower bound for marginal queries applies in the online setting.

To prove the offline upper bound, we use the fact that every query reveals information about other queries. However, we must handle arbitrary sequences of queries, not just the specially-constructed sequences used for the lower bound. The key property of our family of queries is the following. Each element $x$ of the data universe $X$ requires $k$ bits to specify. On the other hand, for any set of queries $q_1, \cdots, q_k$, we can specify $q_1(x), \cdots, q_k(x)$ using only $O(\log(nk))$ bits. Thus the effective size of the data universe given the queries is $\poly(nk)$, rather than $2^k$. Then we can apply a differentially private algorithm that gives good accuracy as long as the data universe has subexponential size~\cite{BlumLR08}. Reducing the size of the data universe is only possible once the queries have been specified; hence this approach only works in the offline setting.


\subsubsection*{Separating Online and Adaptive Queries}
To prove Theorem~\ref{thm:maincorr}, we start with the classical randomized response algorithm~\cite{Warner65}. \jnote{Interestingly, not the oldest paper I've cited. Bonferroni31.} \mnote{Cite something like Markov91, then we'll talk.} Specifically, given a dataset $\db \in \pmo^{n}$, randomized response produces a new dataset $y \in \pmo^{n}$ where each coordinate $y_i$ is independently set to $+x_i$ with probability $ (1+\alpha)/2$ and is set to $-x_i$ with probability $ (1 - \alpha)/2$.  It is easy to prove that this algorithm is $(O(\alpha), 0)$-differentially private.  What accuracy guarantee does this algorithm satisfy? By design, it outputs a vector $y$ that has correlation approximately $\alpha$ with the dataset $\db$ --- that is, $\langle y, x \rangle \approx \alpha n$.  On the other hand, it is also easy to prove that there is no differentially private algorithm (for any reasonable privacy parameters) that can output a vector that has correlation at least $1/2$ with the sensitive dataset.

Our separation between the online and adaptive models is based on the observation that, if we can obtain $O(1/\alpha^2)$ ``independent'' vectors $y_1,\dots,y_k$ that are each roughly $\alpha$-correlated with $\db$, then we can obtain a vector $z$ that is $(1/2)$-correlated with $\db$, simply by letting $z$ be the coordinate-wise majority of the $y_j$s.  Thus, no differentially private algorithm can output such a set of vectors. More precisely, we require that $\langle y_i, y_j \rangle \approx \alpha^2 n$ for $i \ne j$, which is achieved if each $y_j$ is an independent sample from randomized response.

Based on this observation, we devise a class of queries such that, if we are allowed to choose $k$ of these queries adaptively, then we obtain a set of vectors $y_1,\dots,y_k$ satisfying the conditions above. This rules out differential privacy for $k=O(1/\alpha^2)$ adaptive queries. The key is that we can use adaptivity to ensure that each query asks for an ``independent'' $y_j$ by adding the previous answers $y_1, \cdots, y_{j-1}$ as constraints in the search query.

On the other hand, randomized response can answer each such query with high probability. If a number of these queries is fixed in advance, then, by a union bound, the vector $y$ output by randomized response is simultaneously an accurate answer to every query with high probability. Since randomized response is oblivious to the queries, we can also answer the queries in the online model, as long as they are not chosen adaptively.

At a high level, the queries that achieve this property are of the form ``output a vector $y \in \pmo^{n}$ that is approximately $\alpha$-correlated with $\db$ and is approximately as uncorrelated as possible with the vectors $v_1,\dots,v_m$.''  A standard concentration argument shows that randomized response gives an accurate answer to all the queries simulatneously with high probability.  On the other hand, if we are allowed to choose the queries adaptively, then for each query $q_{i}$, we can ask for a vector $y_{i}$ that is correlated with $\db$ but is as uncorrelated as possible with the previous answers $y_1,\dots,y_{i-1}$. 

\subsubsection*{Threshold Queries}
For pure differential privacy, our separation between offline and online threshold queries uses a simple argument based on binary search.  Our starting point is a lower bound showing that any purely differentially private algorithm that takes a dataset of $n$ points $x_1,\dots,x_n \in \set{1,\dots,T}$ and outputs an \emph{approximate median} of these points requires $n = \Omega(\log(T))$.  This lower bound follows from a standard application of the ``packing'' technique of Hardt and Talwar \cite{HardtT10}.  On the other hand, by using binary search, any algorithm that can answer $k = O(\log(T))$ adaptively-chosen threshold queries can be used to find an approximate median. Thus, any purely differentially private algorithm for answering such queries requires a dataset of size $n = \Omega(k)$.  Using the structure of the lower bound argument, we show that the same lower bound holds for online non-adaptive queries as well.  In contrast, using the algorithms of~\cite{DworkNPR10, ChanSS11, DworkNRR15}, we can answer $k$ offline threshold queries on a dataset with only $n = O(\log(k))$ elements, giving an exponential separation.

The basis of our improved algorithm for adaptive threshold queries under approximate differential privacy is a generalization of the \emph{sparse vector} technique \cite{DworkNPR10, RothR10, HardtR10} (see \cite[\S3.6]{DworkR14} for a textbook treatment).  Our algorithm makes crucial use of a \emph{stability argument} similar to the propose-test-release techniques of Dwork and Lei~\cite{DworkL09}.  To our knowledge, this is the first use of a stabiltiy argument for any online or adaptive problem in differential privacy and may be of independent interest.\tnote{I said it.} In particular, our algorithm is given an input $x \in X^n$, a threshold $t \in (0,1)$, and an adaptive sequence of statistical (or low-sensitivity) queries $q_1, \cdots, q_k : X^n \to [0,1]$ and, for each query $q_j$, it reports (i) $q_j(x) \geq t$, (ii) $q_j(x) \leq t$, or (iii) $t-\alpha \leq q_j(x) \leq t+ \alpha$. The sample complexity of this algorithm is $n=O(\sqrt{c} \log(k/\varepsilon\delta)/\varepsilon\alpha)$, where $k$ is the total number of queries, $c$ is an upper bound on the number of times (iii) may be reported, and $(\varepsilon,\delta)$-differential privacy is provided.  We call this the \emph{Between Thresholds algorithm}.

Once we have this algorithm, we can use it to answer adaptively-chosen thresholds using an approach inspired by Bun et al.~\cite{BunNSV15}.  The high-level ideal is to sort the dataset $x_{(1)} < x_{(2)} < \dots < x_{(n)}$ and then partition it into chunks of consecutive sorted elements.  For any chunk, and a threshold $\tau$, we can use the between thresholds algorithm to determine (approximately) whether $\tau$ lies below all elements in the chunk, above all elements in the chunk, or inside the chunk.  Obtaining this information for every chunk is enough to accurately estimate the answer to the threshold query $\tau$ up to an error proportional to the size of the chunks.  The sample complexity is dominated by the $O(\log k)$ sample complexity of our Between Thresholds algorithm multiplied by the number of chunks needed, namely $O(1/\alpha)$.

\jnote{I tried to make it a bit more concrete?  We could also just say something more vague here and move the explanation later where the actual algorithms are in place.}

\section{Preliminaries}

\subsection{Datasets and Differential Privacy}
A \emph{dataset} $\db \in (\row_1,\dots,\row_n) \in X^n$ is an ordered tuple of $n$ elements from some \emph{data universe} $\univ$.  We say that two datasets $\db, \db'$ are \emph{adjacent} if they differ on only a single element and denote this relation by $\db \sim \db'$.
\begin{definition}[Differential Privacy~\cite{DworkMNS06}]
A randomized algorithm $\alg \from \univ^n \to \cR$ is $(\eps, \delta)$-differentially private if for every two adjacent datasets $\db \sim \db'$, and every $R \subseteq \cR$,
$$
\pr{}{\alg(\db) \in R} \leq e^{\eps} \pr{}{\alg(\db') \in R} + \delta.
$$
\end{definition}

We also use the following well known \emph{group privacy} property of $(\eps, 0)$-differential privacy.  We say that two datasets $\db, \db'$ are \emph{$c$-adjacent} if the differ on at most $c$-elements, and denote this relation by $\db \sim_c \db'$.
\begin{lemma}[\cite{DworkMNS06}]
If $\alg \from \univ^n \to \cR$ is $(\eps, 0)$-differentially private, then for every $c \in \N$ and every two $c$-adjacent datasets $\db \sim_c \db'$, and every $R \subseteq \cR$,
$$
\pr{}{\alg(\db) \in R} \leq e^{c \eps} \pr{}{\alg(\db') \in R}.
$$
\end{lemma}

\subsection{Queries}
In this work we consider two general classes of queries on the dataset: \emph{statistical queries}, and \emph{search queries}.  Although statistical queries are a very special case of search queries, we will present each of them independently to avoid having to use overly abstract notation to describe statistical queries.

\paragraph{Statistical Queries.} A statistical query on a data universe $\univ$ is defined by a Boolean predicate $q \from \univ \to \zo$.  Abusing notation, we define the evaluation of a statistical query $q$ on a dataset $\db = (\row_1,\dots,\row_n)$ to be the average of the predicate over the rows
$$
q(\db) = \frac{1}{n} \sum_{i=1}^{n} q(\row_i) \in [0,1].
$$
For a dataset $\db$, a statistical query $q$, and an answer $a \in [0,1]$, the answer is \emph{$\alpha$-accurate for $q$ on $\db$} if $$\left| q(\db) - a \right| \leq \alpha.$$

\paragraph{Search Queries.} A \emph{search query} $q$ on $\univ^n$ is defined by a \emph{loss function} $\loss_q : \univ^n \times \range \to [0,\infty)$, where $\range$ is an arbitrary set representing the range of possible outputs.  For a dataset $\db \in \univ^n$ and an output $y \in \range$, we will say that $y$ is \emph{$\alpha$-accurate for $q$ on $\db$} if $\loss_q(\db, y) \leq \alpha$. In some cases the value of $\loss_q$ will always be either $0$ or $1$. Thus we simply say that $y$ is \emph{accurate for $q$ on $\db$} if $\loss_q(\db,y)=0$.  For example, if $\univ^n = \pmo^n$, we can define a search query by $\range = \pmo^n$, and $\loss_q(\db, y)=0$ if $\langle x, y \rangle \geq \alpha n$ and $\loss_q(\db,y)=1$ otherwise.  In this case, the search query would ask for any vector $y$ that has correlation $\alpha$ with the dataset.

To see that statistical queries are a special case of search queries, given a statistical query $q$ on $\univ^n$, we can define a search query $\loss_q$ with $\range = [0,1]$ and $\loss_q(\db, a)=|q(x)-a|$.  Then both definitions of $\alpha$-accurate align.


\subsection{Models of Interactive Queries}
The goal of this work is to understand the implications of different ways to allow an adversary to query a sensitive dataset.  In each of these models there is an algorithm $\alg$ that holds a dataset $\db \in \univ^n$, and a fixed family of (statistical or search) queries $Q$ on $\univ^n$, and a bound $k$ on the number of queries that $\alg$ has to answer. There is also an adversary $\adv$ that chooses the queries. The models differ in how the queries chosen by $\adv$ are given to $\alg$.

\subsubsection*{Offline}
In the \emph{offline} model, the queries $q_1,\dots,q_k \in Q$ are specified by the adversary $\adv$ in advance and the algorithm $\alg$ is given all the queries at once and must provide answers. Formally, we define the following function $\offlinealg{\adv}{\alg} : \univ^n \to Q^k \times \range^k$ depending $\adv$ and $\alg$.
\begin{figure}[h!]
\begin{framed}
\begin{algorithmic}
\INDSTATE[0]{Input: $\db \in X^n$.}
\INDSTATE[0]{$\adv$ chooses $q_1, \cdots, q_k \in Q$.}
\INDSTATE[0]{$\alg$ is given $\db$ and $q_1, \cdots, q_k$ and outputs $a_1, \cdots, a_k \in \range$.}
\INDSTATE[0]{Output: $(q_1, \cdots, q_k, a_1,\cdots,a_k) \in Q^k \times \range^k$.}
\end{algorithmic}
\end{framed}
\vspace{-6mm}
\caption{$\offlinealg{\adv}{\alg} : \univ^n \to Q^k \times \range^k$}
\end{figure}

\subsubsection*{Online Non-Adaptive}
In the \emph{online non-adaptive} model, the queries $q_1,\dots,q_k \in Q$ are again fixed in advance by the adversary, but are then given to the algorithm one at a time, and the algorithm must give an answer to query $q_j$ before it is shown $q_{j+1}$. We define a function $\onlinealg{\adv}{\alg} : \univ^n \to Q^k \times \range^k$ depending on the adversary $\adv$ and the algorithm $\alg$ as follows.
\mnote{Why not say that $a_j$ may depend on $q_1, \dots, q_{j}$ and any internal state of the mechanism?}
\begin{figure}[h!]
\begin{framed}
\begin{algorithmic}
\INDSTATE[0]{Input: $\db \in X^n$.}
\INDSTATE[0]{$\adv$ chooses $q_1, \cdots, q_k \in Q$.}
\INDSTATE[0]{$\alg$ is given $\db$.}
\INDSTATE[0]{For $j = 1,\dots,k$:}
\INDSTATE[1]{$\alg$ is given $q_j$ and outputs $a_j \in \range$.\footnote{Note that $\alg$ is stateful and its answer $a_j$ will depend on $\db$ and previous queries and answers.}}
\INDSTATE[0]{Output: $(q_1, \cdots, q_k, a_1,\cdots,a_k) \in Q^k \times \range^k$.}
\end{algorithmic}
\end{framed}
\vspace{-6mm}
\caption{$\onlinealg{\adv}{\alg} : \univ^n \to Q^k \times \range^k$}
\end{figure}

\subsubsection*{Online Adaptive}
In the \emph{online adaptive} model, the queries $q_1,\dots,q_k \in Q$ are not fixed, and the adversary may choose each $q_j$ based on the answers that the algorithm gave to the previous queries.  We define a function $\adaptivealg{\adv}{\alg} : \univ^n \to Q^k \times \range^k$ depending on the adversary $\adv$ and the algorithm $\alg$ as follows.
\begin{figure}[h!]
\begin{framed}
\begin{algorithmic}
\INDSTATE[0]{Input: $\db \in X^n$.}
\INDSTATE[0]{$\alg$ is given $\db$.}
\INDSTATE[0]{For $j = 1,\dots,k$:}
\INDSTATE[1]{$\adv$ chooses a query $q_j \in Q$.}
\INDSTATE[1]{$\alg$ is given $q_j$ and outputs $a_j \in \range$.}
\INDSTATE[0]{Output: $(q_1, \cdots, q_k, a_1,\cdots,a_k) \in Q^k \times \range^k$.}
\end{algorithmic}
\end{framed}
\vspace{-6mm}
\caption{$\adaptivealg{\adv}{\alg} : \univ^n \to Q^k \times \range^k$ \label{fig:AdaptiveAlg}}
\end{figure}

\begin{defn}[Differential Privacy for Interactive Mechanisms]
In each of the three cases --- Offline, Online Non-Adaptive, or Online Adaptive ---  we say that $\alg$ is $(\varepsilon,\delta)$-differentially private if, for all adversaries $\adv$, respectively $\offlinealg{\adv}{\alg}$, $\onlinealg{\adv}{\alg}$, or $\adaptivealg{\adv}{\alg}$ is $(\varepsilon,\delta)$-differentially private.
\end{defn}

\begin{defn}[Accuracy for Interactive Mechanisms]
In each case --- Offline, Online Non-Adaptive, or Online Adaptive queries ---  we say that $\alg$ is $(\alpha,\beta)$-accurate if, for all adversaries $\adv$ and all inputs $\db \in \univ^n$, \begin{equation}\pr{q_1, \cdots, q_k, a_1,\cdots,a_k}{\max_{j \in [k]} L_{q_j}(x,a_j) \leq \alpha} \geq 1-\beta,\label{eqn:AccDef}\end{equation} where $(q_1, \cdots, q_k, a_1,\cdots,a_k)$ is respectively drawn from one of $\offlinealg{\adv}{\alg}(\db)$, $\onlinealg{\adv}{\alg}(\db)$, or $\adaptivealg{\adv}{\alg}(\db)$. We also say that $\alg$ is $\alpha$-accurate if the above holds with \eqref{eqn:AccDef} replaced by $$\ex{q_1, \cdots, q_k, a_1,\cdots,a_k}{\max_{j \in [k]} L_{q_j}(x,a_j)} \leq \alpha.$$
\end{defn}

\section{A Separation Between Offline and Online Queries}

In this section we prove that online accuracy is strictly harder to achieve than offline accuracy, even for statistical queries.  We prove our results by constructing a set of statistical queries that we call \emph{prefix queries} for which it is possible to take a dataset of size $n$ and accurately answer superpolynomially many offline prefix queries in a differentially private manner, but it is impossible to answer more than $O(n^2)$ online prefix queries while satisfying differential privacy.

We now define the family of prefix queries.  These queries are defined on the universe $X = \pmo^* = \bigcup_{j=0}^{\infty} \pmo^j$ consisting of all finite length binary strings.\footnote{All of the arguments in this section hold if we restrict to strings of length at most $k+\log n$. However, we allow strings of arbitrary length to reduce notational clutter.}\jnote{Do we really want to make the universe $\pmo^*$?  It's more consistent with the literature to make the universe be $\pmo^d$.  It also highlights the role that the dimension plays in the problem.}\tnote{Perhaps it suffices to have a remark saying we can reduce the universe. Allowing arbirary strings means there are fewer parameters to keep track of.}  For $x, y \in \pmo^*$, we use $y \preceq x$ to denote that $y$ is a \emph{prefix} of $x$.  Formally
$$
y \preceq x \qquad\iff\qquad |y|\leq|x| ~~\text{and} ~~\forall i = 1, \dots, |y|~x_i=y_i.
$$

\begin{definition}
For any finite set $S \subseteq \pmo^*$ of finite-length binary strings, we define the \emph{prefix query} $q_{S} \from \pmo^* \to \pmo$ by
$$
q_{S}(x) = 1 ~~~\iff~~~ \exists y \in S ~~~ y \preceq x.
$$
We also define
\begin{align*}
&\qprefix = \set{q_{S} \mid S \subset \pmo^*} \\
&\qprefix^{B} = \set{q_{S} \mid S \subset \pmo^*, |S| \leq B}
\end{align*}
to be the set of all prefix queries and the set of prefix queries with sizes bounded by $B$, respectively.
\end{definition}

\subsection{Answering Offline Prefix Queries} \label{sec:prefixalg}
We now prove that there is a differentially private algorithm that answers superpolynomially many prefix queries, provided that the queries are specified offline.
\begin{theorem} [Answering Offline Prefix Queries] \label{thm:prefixalg}
For every $\alpha, \eps \in (0, 1/10)$, every $B \in \N$, and every $n \in \N$, there exists a 
$$
k = \min\left\{2^{\Omega(\sqrt{\alpha^3 \eps n})}, 2^{\Omega(\alpha^3 \eps n / \log(B))}\right\}
$$ 
and an $(\eps, 0)$-differentially private algorithm $\alg_{\mathsf{prefix}} \from \univ^n \times (\qprefix^{B})^k \to \R^k$ that is $(\alpha, 1/100)$-accurate for $k$ offline queries from $\qprefix^{B}$.
\end{theorem}
We remark that it is possible to answer even more offline prefix queries by relaxing to $(\eps, \delta)$-differential privacy for some negligibly small $\delta > 0$.  However, we chose to state the results for $(\eps, 0)$-differential privacy to emphasize the contrast with the lower bound, which applies even when $\delta > 0$, and to simplify the statement.

Our algorithm for answering offline queries relies on the existence of a good differentially private algorithm for answering \emph{arbitrary} offline statistical queries.  For concreteness, the so-called ``BLR mechanism'' of Blum, Ligett, and Roth~\cite{BlumLR08} suffices, although different parameter tradeoffs can be obtained using different mechanisms.  Differentially private algorithms with this type of guarantee exist only when the data universe is bounded, which is not the case for prefix queries.  However, as we show, when the queries are specified offline, we can replace the infinite universe $\univ = \pmo^*$ with a finite, restricted universe $\univ'$ and run the BLR mechanism.  Looking ahead, the key to our separation will be the fact that this universe restriction is only possible in the offline setting.  Before we proceed with the proof of Theorem~\ref{thm:prefixalg}, we will state the guarantees of the BLR mechanism.
\begin{theorem}[\cite{BlumLR08}]
For every $0 < \alpha, \eps \leq 1/10$ and every finite data universe $\univ$, if $\cQ_{\mathsf{SQ}}$ is the set of all statistical queries on $\univ$, then for every $n \in \N$, there is a
$$
k = 2^{\Omega(\alpha^3 \eps n / \log |\univ|)}
$$
and an $(\eps, 0)$-differentially private algorithm $\alg_{\mathsf{BLR}} \from \univ^n \times \cQ_{\mathsf{SQ}}^{k} \to \R^k$ that is $(\alpha, 1/100)$-accurate for $k$ offline queries from $\cQ_{\mathsf{SQ}}$.
\end{theorem}

We are now ready to prove Theorem~\ref{thm:prefixalg}.
\begin{proof}[Proof of Theorem~\ref{thm:prefixalg}]
Suppose we are given a set of queries $q_{S_1},\dots,q_{S_k} \in \qprefix^{B}$ and a dataset $\db \in \univ^n$ where $\univ = \pmo^*$.  Let $S = \bigcup_{j=1}^{k} S_{j}$.  We define the universe
$\univ_{S} = S \cup \set{\emptyset}$ where $\emptyset$ denotes the empty string of length $0$.
Note that this universe depends on the choice of queries, and that $|\univ_{S}| \leq kB + 1.$  Since $\univ_{S} \subset \univ$, it will be well defined to restrict the domain of each query $q_{S_{j}}$ to elements of $\univ_{S}$.

Next, given a dataset $\db = (\row_1,\dots,\row_n) \in \univ^n$, and a collection of sets $S_1,\dots,S_k \subset \univ$, we give a procedure for mapping each element of $\db$ to an element of $\univ_{S}$ to obtain a new dataset $\db^{S} = (\row^S_1,\dots,\row^S_n) \in \univ_{S}^n$ that is equivalent to $\db$ with respect to the queries $q_{S_1},\dots,q_{S_k}$.  Specifically, define $r_{S} \from \univ \to \univ_{S}$ by
$$
r_{S}(\row) = \argmax_{y \in \univ_{S}, y \preceq x} |y|.
$$
That is, $r_{S}(\row)$ is the longest string in $\univ_{S}$ that is a prefix of $\row$.  We summarize the key property of $r_{S}$ in the following claim
\begin{claim} \label{clm:reduction}
For every $\row \in \univ$, and $j = 1,\dots,k$, $q_{S_j}(r_{S}(\row)) = q_{S_j}(\row)$.
\end{claim}
\begin{proof}[Proof of Claim~\ref{clm:reduction}]
First, we state a simple but important fact about prefixes:  If $y, y'$ are both prefixes of a string $x$ with $|y| \le |y'|$, then $y$ is a prefix of $y'$.  Formally,
\begin{equation}\label{eq:prefixfact}
\forall x,y,y' \in \{0,1\}^* \qquad (y \preceq x ~~\wedge~~ y' \preceq x ~~\wedge~~ |y| \le |y'|) ~~ \implies ~~ y \preceq y' .
\end{equation}
Now, fix any $\row \in \univ$ and any query $q_{S_j}$ and suppose that $q_{S_j}(\row) = 1$.  Then there exists a string $y \in S_j$ such that $y \preceq \row$.  By construction, we have that $r_{S}(\row) \preceq x$ and that $|r_{S}(\row)| \geq |y|$.  Thus, by~\eqref{eq:prefixfact}, we have that $y \preceq r_{S}(\row)$.  Thus, there exists $y \in S_j$ such that $y \preceq r_{S}(\row)$, which means $q_{S_j}(r_{S}(\row)) = 1$, as required.

Next, suppose that $q_{S_j}(r_{S}(\row)) = 1$.  Then, there exists $y \in S_j$ such that $y \preceq r_{S}(\row)$.  By construction, $r_{S}(\row) \preceq \row$, so by transitivity we have that $y \preceq \row$.  Therefore, $q_{S_j}(\row) = 1$, as required.
\end{proof}

Given this lemma, we can replace every row $\row_i$ of $\db$ with $\row^S_i = r_{S}(\row_i)$ to obtain a new dataset $\db^{S}$ such that for every $j = 1,\dots,k$,
$$
q_{S_j}(\db^{S}) = \frac{1}{n} \sum_{i=1}^{n} q_{S_j}(\row^S_i) = \frac{1}{n} \sum_{i=1}^{n} q_{S_j}(\row_i) = q_{S_j}(\db).
$$
Thus, we can answer $q_{S_1}, \cdots, q_{S_k}$ on $\db^S \in \univ_S^n$, rather than on $\db \in \univ^n$. Note that each row of $\db^{S}$ depends only on the corresponding row of $\db$.  Hence, for every set of queries $q_{S_1},\dots,q_{S_k}$, if $\db \sim \db'$ are adjacent datasets, then $\db^{S} \sim \db'^{S}$ are also adjacent datasets. Consequently, applying a $(\varepsilon,\delta)$-differentially private algorithm to $\db^S$ yields a $(\varepsilon,\delta)$-differentially private algorithm as a function of $\db$.

In particular, we can give $\alpha$-accurate answers to these queries using the algorithm $\alg_{\mathsf{BLR}}$ as long as 
$$
k \leq 2^{\Omega(\alpha^3 \eps n / \log |\univ_{S}|)} = 2^{\Omega(\alpha^3 \eps n / \log(kB+1))}.
$$
Rearranging terms gives the bound in Theorem~\ref{thm:prefixalg}.  We specify the complete algorithm $\alg_{\mathsf{prefix}}$ in Figure \ref{alg:prefix}.
\begin{figure}[h!]
\begin{framed}
\begin{algorithmic}
\INDSTATE[0]{$\alg_{\mathsf{prefix}}(\db; q_{S_1},\dots,q_{S_k})$:}
\INDSTATE[1]{Write $\db = (\row_1,\dots,\row_n) \in \univ^n$, $S = \bigcup_{j=1}^{k} S_j$, $\univ_S = S \cup \set{\emptyset}$.}
\INDSTATE[1]{For $i = 1,\dots,n$, let $\row^S_i = r_{S}(\row_i)$ and let $\db^{S} = (\row^S_1,\dots,\row^S_n) \in \univ_S^n$.}
\INDSTATE[1]{Let $(a_1,\dots,a_k) = \alg_{\mathsf{BLR}}(\db^{S}; q_{S_1},\dots,q_{S_k})$.}
\INDSTATE[1]{Output $(a_1,\dots,a_k)$.}
\end{algorithmic}
\end{framed}
\vspace{-6mm}
\caption{$\alg_{\mathsf{prefix}}$} \label{alg:prefix}
\end{figure}

\end{proof}

\subsection{A Lower Bound for Online Prefix Queries} \label{sec:prefixlb}
Next, we prove a lower bound for online queries.  Our lower bound shows that the simple approach of perturbing the answer to each query with independent noise is essentially optimal for prefix queries.  Since this approach is only able to answer $k = O(n^2)$ queries, we obtain an exponential separation between online and offline statistical queries for a broad range of parameters.

\begin{theorem}[Lower Bound for Online Prefix Queries] \label{thm:prefixlb}
There exists a function $k = O(n^2)$ such that for every sufficiently large $n \in \N$, there is no $(1, 1/30n)$-differentially private algorithm $\alg$ that takes a dataset $\db \in \univ^n$ and is $(1/100, 1/100)$-accurate for $k$ online queries from $\qprefix^{n}$. 
\end{theorem}
In this parameter regime, our algorithm from Section~\ref{sec:prefixalg} answers $k=\exp(\tilde{\Omega}(\sqrt{n}))$ offline prefix queries, so we obtain an exponential separation.

Our lower bound relies on a connection between \emph{fingerprinting codes} and differential privacy \cite{Ullman13, BunUV14, SteinkeU15b, DworkSSUV15}.  However, instead of using fingerprinting codes in a black-box way, we will make a direct use of the main techniques.  Specifically, we will rely heavily on the following key lemma. The proof appears in Appendix \ref{sec:Fingerprinting}.	

\begin{lem}[Fingerprinting Lemma] \label{lem:Fingerprinting} \tnote{Somehow having $c$ be a random variable still seems unnatural to me...}\jnote{True, random variables should be upper case.  We could change it, or at least write $p \sim [-1,1], c \sim C_p$ to make it clear that $c$ is the realization not the rv.}
Let $f : \pmo^n \to [-1,1]$ be any function.  Suppose $p$ is sampled from the uniform distribution over $[-1,1]$ and $c \in \pmo^n$ is a vector of $n$ independent bits, where each bit has expectation $p$.  Letting $\overline{c}$ denote the coordinate-wise mean of $c$, we have
$$\ex{p,c}{f(c)\cdot \sum_{i \in [n]}(c_i-p) + 2\big| f(c) - \overline{c} \, \big|} \geq \frac{1}{3}.$$
\end{lem}
Roughly the fingerprinting lemma says that if we sample a vector $c \in \pmo^n$ in a specific fashion, then for any bounded function $f(c)$, we either have that $f(c)$ has ``significant'' correlation with $c_i$ for some coordinate $i$, or that $f(c)$ is ``far'' from $\overline{c}$ on average.  In our lower bound, the vector $c$ will represent a \emph{column} of the dataset, so each coordinate $c_i$ will correspond to the value of some \emph{row} of the dataset.  The function $f(c)$ will represent the answer to some prefix query.  We will use the accuracy of a mechanism for answering prefix queries to argue that $f(c)$ is not far from $\overline{c}$, and therefore conclude that $f(c)$ must be significantly correlated with some coordinate $c_i$.  On the other hand, if $c_i$ were excluded from the dataset, then $c_i$ is sufficiently random that the mechanism's answers cannot be significantly correlated with $c_i$.  We will use this to derive a contradiction to differential privacy.

\begin{proof}[Proof of Theorem~\ref{thm:prefixlb}]
First we define the distribution on the input dataset $\db = (\row_1,\dots,\row_n)$ and the queries $q_{S_1}, \cdots, q_{S_k}$.
\paragraph{Input dataset $\db$:}
\begin{itemize}
\item Sample $p^1, \cdots, p^k \in [-1,1]$ independently and uniformly at random.
\item Sample $c^1, \cdots, c^k \in \pmo^n$ independently, where each $c^j$ is a vector of $n$ independent bits, each with expectation $p^j$.
\item For $i \in [n]$, define 
$$
\row_i = (\mathrm{binary}(i), c_i^1, \cdots, c_i^k) \in \pmo^{\lceil\log_2 n\rceil+k},
$$
where $\mathrm{binary}(i) \in \pmo^{\lceil \log_2 n \rceil}$ is the binary representation of $i$ where $1$ is mapped to $+1$ and $0$ is mapped to $-1$.\footnote{This choice is arbitrary, and is immaterial to our lower bound.  The only property we need is that $\mathrm{binary}(i)$ uniquely identifies $i$ and, for notational consistency, we require $\mathrm{binary}(i)$ to be a string over the alphabet $\pmo$.}  Let $\db = (\row_1,\dots,\row_n) \in \left(\pmo^{\lceil \log_2 n\rceil + k}\right)^n$.
\end{itemize}

\paragraph{Queries $q_{S_1}, \cdots, q_{S_k}$:}
\begin{itemize}
\item For $i \in [n]$ and $j \in [k]$, define 
$$
z_{i,j} = (\mathrm{binary}(i),c_i^1, \cdots, c_i^{j-1},1) \in \pmo^{\lceil \log_2 n \rceil + j}.
$$ 
\item For $j \in [k]$, define $q_{S_j} \in \qprefix^{n}$ by $S_j = \set{z_{i,j} \mid i \in [n]}$.
\end{itemize}

These queries are designed so that the correct answer to each query $j \in [k]$ is given by $q_{S_j}(\db)=\overline{c}^{j}$:
\begin{claim} \label{clm:Queries}
For every $j \in [k]$, if the dataset $\db$ and the queries $q_{S_1},\dots,q_{S_k}$ are constructed as above, then with probability $1$,
$$
q_{S_j}(\db) = \frac{1}{n} \sum_{i=1}^{n} q_{S_j}(\row_i) = \frac{1}{n} \sum_{i=1}^{n} c^j_i = \overline{c}^j
$$
\end{claim}
\begin{proof}[Proof of Claim~\ref{clm:Queries}]
We have $$q_{S_j}(\row_i) = 1 \iff \exists w \in S_j ~(w \preceq \row_i) \iff \exists \ell \in [n] ~(z_{\ell,j} \preceq \row_i).$$ By construction, we have $z_{\ell,j} \preceq \row_i$ if and only if $\ell=i$ and $x_i^j=c_i^j = 1$, as required.  Here, we have used the fact that the strings $\mathrm{binary}(i)$ are unique to ensure that $z_{\ell, j} \preceq \row_i$ if and only if $\ell = i$.
\end{proof}

We now show no differentially private algorithm $\alg$ is capable of giving accurate answers to these queries. Let $\alg$ be an algorithm that answers $k$ online queries from $\qprefix^{n}$.  Suppose we generate an input dataset $\db$ and queries $q_{S_1},\dots,q_{S_k}$ as above, and run $\alg(\db)$ on this sequence of queries.  Let $a^1,\dots,a^k \in [-1,1]$ denote the answers given by $\alg$.

First, we claim that, if $\alg(\db)$ is accurate for the given queries, then each answer $a^{j}$ is close to the corresponding value $\overline{c}^{j} = \frac{1}{n} \sum_{i=1}^{n} c^{j}_{i}$.
\begin{clm} \label{clm:Accuracy}
If $\alg$ is $(1/100,1/100)$-accurate for $k$ online queries from $\qprefix^{n}$, then with probability $1$ over the choice of $\db$ and $q_{S_1},\dots,q_{S_k}$ above,
$$
\ex{\alg}{\sum_{j \in [k]} \left| a^j - \overline{c}^j \right|} \leq \frac{k}{10}. 
$$
\end{clm}
\begin{proof}[Proof of Claim~\ref{clm:Accuracy}]
By Claim~\ref{clm:Queries}, for every $j \in [k]$, $q_{S_j}(\db) = \overline{c}^{j}$.
Since, by assumption, $\alg$ is $(1/100, 1/100)$-accurate for $k$ online queries from $\qprefix^{n}$, we have that with probability at least $99/100$,
$$
\forall j \in [k]~~\left| a^{j} - q_{S_j}(\db) \right| \leq \frac{1}{100} ~~ \implies ~~ \forall j \in [k]~~\left| a^{j} - \overline{c}^{j} \right| \leq \frac{1}{100}
$$
By linearity of expectation, this case contributes at most $k/100$ to the expectation.  On the other hand, $| a^{j} - q_{S_j}(\db) | \leq 2$, so by linearity of expectation the case where $\alg$ is inaccurate contributes at most $2k/100$ to the expectation.  This suffices to prove the claim.
\end{proof}

The next claim shows how the fingerprinting lemma (Lemma~\ref{lem:Fingerprinting}) can be applied to $\alg$.
\begin{clm} \label{clm:Fingerprinting}
$$\ex{p,\db,q,\alg}{\sum_{j \in [k]} \left(a^j \sum_{i \in [n]} (c_i^j - p^j) + 2\left| a^j - \overline{c}^j \right|\right)} \geq \frac{k}{3}. $$
\end{clm}
\begin{proof}
By linearity of expectation, it suffices to show that, for every $j \in [k]$, 
$$
\ex{p,\db,q,\alg}{a^j \sum_{i \in [n]} (c_i^j - p^j) + 2\left| a^j - \overline{c}^j \right|} \geq \frac{1}{3}.
$$
Since each column $c^{j}$ is generated independently from the columns $c^{1},\dots,c^{j-1}$, $c^j$ and $p^j$ are independent from $q_{S_1}, \cdots, q_{S_j}$. Thus, at the time $\alg$ produces the output $a^j$, it does not have any information about $c^j$ or $p^j$ apart from its private input.  (Although $\alg$ later learns $c^j$ when it is asked $q_{S_{j+1}}$.) For any fixed values of $c^1, \dots, c^{j-1}$ and the internal randomness of $\alg$, the answer $a^j$ is a deterministic function of $c^j$. Thus we can apply Lemma \ref{lem:Fingerprinting} to this function to establish the claim.
\end{proof}

Combining Claims~\ref{clm:Accuracy} and~\ref{clm:Fingerprinting} gives
$$
\ex{p,\db,q,\alg}{\sum_{j \in [k]} a^j \sum_{i \in [n]} (c_i^j - p^j)} \geq \frac{2k}{15}.
$$
In particular, there exists some $i^* \in [n]$ such that \mnote{Replaced $i$ with $i^*$}
\begin{equation} \label{eq:highcorrelationwithi}
\ex{p,\db,q,\alg}{\sum_{j \in [k]} a^j  (c_{i^*}^j - p^j)} \geq \frac{2k}{15n}.
\end{equation}
To complete the proof, we show that \eqref{eq:highcorrelationwithi} violates the differential privacy guarantee unless $n \geq \Omega(\sqrt{k})$.

To this end, fix any $p^1,\dots,p^k \in [-1,1]$, whence ${c}^1_{i^*}, \cdots, {c}^k_{i^*} \in \pmo$ are independent bits with $\ex{}{{c}^j}=p^j$. Let $\tilde{c}^1, \cdots, \tilde{c}^k \in \pmo$ be independent bits with $\ex{}{\tilde{c}^j}=p^j$. The random variables ${c}^1_{i^*}, \cdots, {c}^k_{i^*}$ have the same marginal distribution as $\tilde{c}^1, \cdots, \tilde{c}^k$. However, $\tilde{c}^1, \cdots, \tilde{c}^k$ are independent from $a^1, \cdots, a^k$, whereas $a^1, \cdots, a^k$ depend on $c_{i^*}^1, \cdots, c_{i^*}^k$. Consider the quantities   \mnote{Replaced all instances of $Z_i$ with $Z$}
$$
Z = \sum_{j \in [k]} a^j  (c_{i^*}^j - p^j) \qquad \text{and} \qquad \tilde{Z} = \sum_{j \in [k]} a^j  (\tilde{c}^j - p^j).
$$
Differential privacy implies that $Z$ and $\tilde{Z}$ have similar distributions.  Specifically, if $\alg$ is $(1,1/30n)$-differentially private, then
$$
\ex{}{|Z|} = \int_0^{2k} \pr{}{|Z|>z} \mathrm{d}z
\leq \int_0^{2k} \left(e\pr{}{|\tilde{Z}|>z} + \frac{1}{30n}\right) \mathrm{d}z
= e \ex{}{|\tilde{Z}|} + \frac{k}{15n},
$$
as $|Z|, |\tilde{Z}| \leq 2k$ with probability 1.

Now $\ex{}{|Z|} \geq \ex{}{Z} \geq 2k/15n$, by \eqref{eq:highcorrelationwithi}. On the other hand, $a^j$ is independent from $\tilde{c}^j$ and $\ex{}{\tilde{c}^j-p^j}=0$, so $\ex{}{\tilde{Z}}=0$. We now observe that $$\ex{}{|\tilde{Z}|}^2 \leq \ex{}{\tilde{Z}^2}=\var{}{\tilde{Z}} = \sum_{j \in [k]} \var{}{a^j (\tilde{c}^j-p^j)} \leq  \sum_{j \in [k]} \ex{}{(\tilde{c}^j-p^j)^2}  \leq k.$$ Thus, we have
$$
\frac{2k}{15n} \leq \ex{}{|Z|} \leq e \ex{}{|\tilde{Z}|} + \frac{k}{15n} \leq e \sqrt{k} + \frac{k}{15n}.
$$
The condition $2k/15n \leq e \sqrt{k} + k/15n$ is a contradiction unless $k \leq 225e^2 n^2$.  Thus, we can conclude that there exists a $k = O(n^2)$ such that no $(1, 1/30n)$-differentially private algorithm is accurate for more than $k$ online queries from $\qprefix^{n}$, as desired.  This completes the proof.
\end{proof}

\section{A Separation Between Adaptive and Non-Adaptive Online Queries}
In this section we prove that even among online queries, answering adaptively-chosen queries can be strictly harder than answering non-adaptively-chosen queries.  Our separation applies to a family of search queries that we call \emph{correlated vector queries.}  We show that for a certain regime of parameters, it is possible to take a dataset of size $n$ and privately answer an exponential number of fixed correlated vector queries, even if the queries are presented online, but it is impossible to answer more than a constant number of adaptively-chosen correlated vector queries under differential privacy.

The queries are defined on datasets $\db \in \pmo^{n}$ (hence the data universe is $\univ = \pmo$).  For every query, the range $\range = \pmo^{n}$ is the set of $n$-bit vectors.  We fix some parameters $0 < \alpha < 1$ and $m \in \N$.  A query $q$ is specified by a set $V$ where $V = \set{v^1,\dots,v^m} \subseteq \pmo^{n}$ is a set of $n$-bit vectors.  Roughly, an accurate answer to a given search query is any vector $y \in \pmo^{n}$ that is approximately $\alpha$-correlated with the input dataset $\db \in \pmo^{n}$ and has nearly as little correlation as possible with every $v^{j}$.  By ``as little correlation as possible with $v^{j}$'' we mean that $v^{j}$ may itself be correlated with $\db$, in which case $y$ should be correlated with $v^{j}$ only insofar as this correlation comes through the correlation between $y$ and $\db$.  \mnote{I didn't understand the previous description. Is this statement still accurate?} Formally, for a query $q_{V}$, we define the loss function $\loss_{q_{V}} : \univ^n \times \univ^n \to \{0,1\}$ by
$$
\loss_{q_{V}}(\db, y) = 0 ~~~\iff~~~ {\left| \langle y - \alpha x, x \rangle \right| \leq \frac{\alpha^2 n}{100}~~\land~~\forall v^{j} \in V\; \left| \langle y - \alpha x, v^j \rangle \right| \leq \frac{\alpha^2 n}{100}}.
$$
We remark that the choice of $\alpha^2 n /100$ is somewhat arbitrary, and we can replace this choice with $C$ for any $\sqrt{n} \ll C \ll n$ and obtain quantitatively different results.  We chose to fix this particular choice in order to reduce notational clutter.
\jnote{$C := \alpha^2 \sqrt{n}/100$ since that's the parameter regime of interest for the lower bound.}
We let
$$
\qcorr^{n, \alpha, m} = \set{q_{V} \mid V \subseteq \pmo^n, |V| \leq m}
$$
be the set of all correlated vector queries on $\pmo^{n}$ for parameters $\alpha, m$.
\jnote{We might want to compile all of this into a defn environment.}

\subsection{Answering Online Correlated Vector Queries} \label{sec:corralg}
Provided that all the queries are fixed in advance, we can privately answer correlated vector queries using the randomized response algorithm.  This algorithm simply takes the input vector $\db \in \pmo^{n}$ and outputs a new vector $y \in \pmo^{n}$ where each bit $y_i$ is independent and is set to $x_i$ with probability $1/2 + \rho$ for a suitable choice of $\rho > 0$.  The algorithm will then answer every correlated vector query with this same vector $y$.  The following theorem captures the parameters that this mechanism achieves.
\begin{theorem}[Answering Online Correlated Vector Queries]\label{thm:corralg}
For every $0 < \alpha < 1/2$, there exists $k = 2^{\Omega(\alpha^4 n)}$ such that, for every sufficiently large $n \in \N$, there is a $(3\alpha, 0)$-differentially private algorithm $\alg_{\mathsf{corr}}$ that takes a dataset $\db \in \pmo^n$ and is $(1/k)$-accurate for $k$ online queries from $\qcorr^{n,\alpha,k}$.
\end{theorem}
\jnote{The quantification in this theorem is really ugly.}
\begin{proof}[Proof Theorem~\ref{thm:corralg}]
Our algorithm based on randomized response is presented in Figure \ref{fig:rr} below.
\begin{figure}[h!]
\begin{framed}
\begin{algorithmic}
\INDSTATE[0]{$\alg_{\mathsf{corr}}$:}
\INDSTATE[1]{Input: a dataset $\db \in \pmo^{n}$.}
\INDSTATE[1]{Parameters: $0 < \alpha < 1/2$.} 
\INDSTATE[1]{For $i = 1,\dots,n$, independently set
\begin{equation*}
y_i =
\begin{cases}
+\row_i &\textrm{with probability $\frac{1+\alpha}{2}$}\\
-\row_i &\textrm{with probability $\frac{1-\alpha}{2}$}
\end{cases}.
\end{equation*}}
\INDSTATE[1]{Let $y = (y_1,\dots,y_n) \in \pmo^{n}$, and answer each query with $y$.}
\end{algorithmic}
\end{framed}
\vspace{-6mm}
\caption{$\alg_{\mathsf{corr}}$} \label{fig:rr}
\end{figure}

To establish privacy, observe that by construction each output bit $y_i$ depends only on $x_i$ and is independent of all $x_j, y_j$ for $j \ne i$.  Therefore, it suffices to observe that if $0 < \alpha < 1/2$,
$$
1\leq \frac{\mathbb{P}[y_i = +1 \mid x_i = +1]}{\mathbb{P}[y_i = +1 \mid x_i = -1]} = \frac{1+\alpha}{1-\alpha} \leq e^{3 \alpha}
$$
and similarly
$$
1\geq \frac{\mathbb{P}[y_i = -1 \mid x_i = +1]}{\mathbb{P}[y_i = -1 \mid x_i = -1]} = \frac{1-\alpha}{1+\alpha} \geq e^{-3\alpha}.
$$

To prove accuracy, observe that since the output $y$ does not depend on the sequence of queries, we can analyze the mechanism as if the queries $q_{V_1},\dots,q_{V_k} \in \qcorr^{n,\alpha,k}$ were fixed and given all at once.  Let $V = \bigcup_{j=1}^{k} V_{j}$, and note that $|V| \leq k^2$.  First, observe that $\ex{}{y} = \alpha x$.  Thus we have
$$
\ex{y}{\langle y - \alpha x, x \rangle} = 0~~~~\textrm{and}~~~~\forall v \in V~~\ex{y}{\langle y - \alpha x, v\rangle} = 0
$$
Since $x$ and every vector in $V$ is fixed independently of $y$, and the coordinates of $y$ are independent by construction, the quantities $\langle y , x \rangle$ and $\langle y , v \rangle$ are each the sum of $n$ independent $\pmo$-valued random variables.  Thus, we can apply Hoeffding's inequality\footnote{We use the following statement of Hoeffding's Inequality: if $Z_1,\dots,Z_n$ are independent $\pmo$-valued random variables, and ${Z} = \sum_{i=1}^{n} Z_i$, then $$\pr{}{\left|~{Z} - \ex{}{{Z}}~\right| > C\sqrt{n}} \leq 2e^{-C^2/2}$$} and a union bound to conclude
\begin{align*}
&\pr{y}{|\langle y - \alpha x, x \rangle| > \frac{\alpha^2 n}{100}} \leq 2\exp\left(\frac{-\alpha^4 n}{20000}\right)\\
&\pr{y}{\exists v \in V\textrm{ s.t. }|\langle y - \alpha x, v \rangle| > \frac{\alpha^2 n}{100}} \leq 2k^2\exp\left(\frac{-\alpha^4 n}{20000}\right)
\end{align*}
The theorem now follows by setting an appropriate choice of $k = 2^{\Omega(\alpha^4 n)}$ such that $2(k^2+1)\cdot\exp\left(\frac{-\alpha^4 n}{20000}\right) \leq 1/k$.
\end{proof}

\subsection{A Lower Bound for Adaptive Correlated Vector Queries} \label{sec:corrlb}
We now prove a contrasting lower bound showing that if the queries may be chosen adaptively, then no differentially private algorithm can answer more than a constant number of correlated vector queries.  The key to our lower bound is that fact that adaptively-chosen correlated vector queries allow an adversary to obtain many vectors $y^1,\dots,y^k$ that are correlated with $\db$ but pairwise nearly orthogonal with each other.  As we prove, if $k$ is sufficiently large, this information is enough to recover a vector $\tilde{x}$ that has much larger correlation with $\db$ than any of the vectors $y^1,\dots,y^k$ have with $\db$.  By setting the parameters appropriately, we will obtain a contradiction to differential privacy.

\begin{theorem}[Lower Bound for Correlated Vector Queries] \label{thm:corrlb}
For every $0 < \alpha < 1/2$, there is a $k = O(1/\alpha^2)$ such that for every sufficiently large $n \in \N$, there is no $(1,1/20)$-differentially private algorithm that takes a dataset $\db \in \pmo^{n}$ and is $1/100$-accurate for $k$ adaptive queries from $\qcorr^{n,\alpha,k}$
\end{theorem}
We remark that the value of $k$ in our lower bound is optimal up to constants, as there is a $(1, 1/20)$-differentially private algorithm that can answer $k = \Omega(1/\alpha^2)$ adaptively-chosen queries of this sort.  The algorithm simply answers each query with an independent invocation of randomized response.  Randomized response is $O(\alpha)$-differentially private for each query, and we can invoke the adaptive composition theorem~\cite{DworkMNS06, DworkRV10} to argue differential privacy for $k = \Omega(1/\alpha^2)$-queries.\jnote{Not sure that this remark is interesting.}

Before proving Theorem \ref{thm:corrlb}, we state and prove the combinatorial lemma that forms the foundation of our lower bound.
\begin{lem}[Reconstruction Lemma]\label{lem:reconstruction}
Fix parameters $0 \leq a,b \leq 1$.  Let $\db \in \pmo^{n}$ and $y^1, \cdots, y^k \in \pmo^{n}$ be vectors such that
\begin{align*}
&\forall 1 \leq j \leq k ~~~~ \langle y^j , x \rangle \geq a n\\
&\forall 1 \leq j < j' \leq k ~~~~ |\langle y^{j}, y^{j'} \rangle| \leq b n.
\end{align*}
Then, if we let $\tilde{\db} = \mathrm{sign}(\sum_{j=1}^{k} y^j) \in \pmo^{n}$ be the coordinate-wise majority of $y^{1},\dots,y^{k}$, we have
$$
\langle \tilde{\db}, x \rangle \geq \left(1-\frac{2}{a^2 k} - \frac{2(b-a^2)}{a^2} \right) n.
$$
\end{lem}
\begin{proof}[Proof of Lemma~\ref{lem:reconstruction}]
Let 
$$
\overline{y} = \frac{1}{k} \sum_{j=1}^k y^j \in [-1,1]^{n}.
$$
By linearity, $\langle \overline{y}, x \rangle \geq a n$ and
$$
\|\overline{y}\|_2^2 = \frac{1}{k^2} \sum_{j, j' = 1}^{k} \langle y^{j}, y^{j'} \rangle \leq \frac{1}{k^2} \left( k n + (k^2-k) b n\right) \leq \left(\frac{1}{k} + b \right) n.
$$
Define a random variable $W \in [-1,1]$ to be $x_i \overline{y}_i$ for a uniformly random $i \in [n]$. Then 
$$
\ex{}{W}= \frac{1}{n} \langle x, \overline{y} \rangle \geq a~~~\textrm{and}~~~\ex{}{W^2} = \frac{1}{n} \sum_{i=1}^{n} x_i^2 \overline{y}_i^2 = \frac{1}{n} \|\overline{y}\|_2^2 \leq {\frac{1}{k} + b}
$$
By Chebyshev's inequality, 
$$
\pr{}{W \leq 0} \leq \pr{}{|W-\ex{}{W}| \geq a} \leq \frac{\operatorname{Var}[W]}{a^2} = \frac{\mathbb{E}[W^2]-\mathbb{E}[W]^2}{a^2} \leq \frac{\frac{1}{k} + b - a^2}{a^2}.
$$
Meanwhile,
$$
\pr{}{W \leq 0} 
= \frac{1}{n} \sum_{i=1}^n \mathbb{I}[x_i \overline{y}_i \leq 0] 
\geq  \frac{1}{n} \sum_{i=1}^n \mathbb{I}[\mathrm{sign}(\overline{y}_i) \ne x_i] 
= \frac12-\frac{1}{2n} \langle \mathrm{sign}(\overline{y}), x \rangle.
$$
Thus we conclude
$$
\langle \mathrm{sign}(\overline{y}), x \rangle \geq n - 2n \pr{}{W \leq 0} \geq n - 2n\left( \frac{\frac{1}{k} + b-a^2}{a^2}\right)
$$
To complete the proof, we rearrange terms and note that $\mathrm{sign}(\overline{y}) = \mathrm{sign}(\sum_{j=1}^{k} y^{j})$.
\end{proof}

Now we are ready to prove our lower bound for algorithms that answer adaptively-chosen correlated vector queries.

\begin{proof}[Proof of Theorem~\ref{thm:corrlb}]
We will show that the output $y^{1},\dots,y^{k}$ of any algorithm $\alg$ that takes a dataset $\db \in \pmo^{n}$ and answers $k = 100/\alpha^2$ adaptively-chosen correlated vector queries can be used to find a vector $\tilde{\db} \in \pmo^{n}$ such that $\langle \tilde{\db}, \db \rangle > n/2$.  In light of Lemma~\ref{lem:reconstruction}, this vector will simply be $\tilde{\db} = \mathrm{sign}(\sum_{j=1}^{k} y^{j})$.  We will then invoke the following elementary fact that differentially private algorithms do not admit this sort of reconstruction of their input dataset.
\begin{fact} \label{fact:noreconstruction}
For every sufficiently large $n \in \N$, there is no $(1,1/20)$-differentially private algorithm $\alg \from \pmo^{n} \to \pmo^{n}$ such that for every $\db \in \pmo^{n}$, with probability at least $99/100$, $\langle \alg(\db), \db \rangle > n/2$.
\end{fact}

The attack works as follows.  For $j = 1,\dots,k$, define the set $V_j = \set{y^{1},\dots,y^{j-1}}$ and ask the query $q_{V_j}(\db) \in \qcorr^{n, \alpha, k}$ to obtain some vector $y^{j}$.  Since $\alg$ is assumed to be accurate for $k$ adaptively-chosen queries, with probability $99/100$, we obtain vectors $y^{1},\dots,y^{k} \in \pmo^{n}$ such that \jnote{Might want to spell this calculation out a bit more for people who forgot our def of accuracy.}  \mnote{Done, also corrected some indexing errors}
\begin{align*}
\forall 1 \leq j \leq k ~~~~ \langle y^j , \db \rangle &\ge \langle \alpha \db, \db \rangle - |\langle y - \alpha \db, \db \rangle| \\
 &\geq \alpha n - \frac{\alpha^2n}{100} \\
 &\geq a n,
\end{align*}
\begin{align*}
\forall 1 \leq j < j' \leq k ~~~~ |\langle y^{j}, y^{j'} \rangle| &\leq |\langle \alpha \db, y^{j} \rangle| + |\langle y^{j'} - \alpha \db, y^{j}\rangle|\\
& \leq \alpha|\langle y^{j}, \db \rangle| + \frac{\alpha^2 n}{100} \\
& \leq \alpha\left( |\langle \alpha \db, \db \rangle| + |\langle y^j - \alpha\db, \db \rangle| \right) + \frac{\alpha^2 n}{100} \\
& \leq \alpha^2 n + \frac{\alpha^3 n}{100} + \frac{\alpha^2 n}{100} \\
&\leq \frac{51}{50}\alpha^2n \\
&= b n,
\end{align*}
where $a = 99\alpha/100$ and $b = 51 \alpha^2/50$.
Thus, by Lemma~\ref{lem:reconstruction}, if $\tilde{\db} = \mathrm{sign}(\sum_{j=1}^{k} y^{j})$, and $k = 100/\alpha^2$, we have
\begin{align*}
\langle \tilde{\db}, \db \rangle 
&\geq  \left(1-\frac{2}{a^2 k} - \frac{2(b-a^2)}{a^2} \right) n\\
&= \left(1 - \frac{2}{(99\alpha/100)^2 k} - \frac{2(51\alpha^2/50-(99\alpha/100)^2)}{(99\alpha/100)^2} \right)n \\
&= \left(1 - \frac{2(100/99)^2}{100} - 2\left(\frac{(51/50)-(99/100)^2}{(99/100)^2}\right)\right)n \\
&\geq 0.89n \geq n/2.
\end{align*}
By Fact~\ref{fact:noreconstruction}, this proves that $\alg$ cannot be $(1, 1/20)$-differentially private.
\end{proof}

\section{Threshold Queries}

First we define threshold queries, which are a family of statistical queries.

\begin{defn}
Let $\thresh{X}$ denote the class of threshold queries over a totally ordered domain $X$. That is, $\thresh{X} = \{c_x : x \in X\}$ where $c_x : X \to \{0, 1\}$ is defined by $c_x(y) = 1$ iff $y \le x$.
\end{defn}

\subsection{Separation for Pure Differential Privacy}

In this section, we show that the sample complexity of answering adaptively-chosen thresholds can be exponentially larger than that of answering thresholds offline.

\begin{prop}[\cite{DworkNPR10, ChanSS11, DworkNRR15}]
Let $X$ be any totally ordered domain. Then there exists a $(\varepsilon,0)$-differentially private mechanism $M$ that, given $\db \in X^n$, gives $\alpha$-accurate answers to $k$ offline queries from $\thresh{X}$ for 
\[n = O \left(\min\left\{\frac{\log k + \log^2(1/\alpha)}{\alpha\varepsilon}, \frac{\log^2 k}{\alpha\varepsilon} \right\}\right)\]
\end{prop}

On the other hand, we show that answering $k$ adaptively-chosen threshold queries can require sample complexity as large as $\Omega(k)$ -- an exponential gap. Note that this matches the upper bound given by the Laplace mechanism \cite{DworkMNS06}. 

\begin{prop}\label{prop:thresh_adaptive_lb}
Answering $k$ adaptively-chosen threshold queries on $[2^{k-1}]$ to accuracy $\alpha$ subject to $\varepsilon$-differential privacy requires sample complexity $n=\Omega(k/\alpha\varepsilon)$.
\end{prop}

The idea for the lower bound is that an analyst may adaptively choose $k$ threshold queries to binary search for an ``approximate median'' of the dataset. However, a packing argument shows that locating an approximate median requires sample complexity $\Omega(k)$.
\begin{defn}[Approximate Median]
Let $X$ be a totally ordered domain, $\alpha>0$, and $x \in X^n$. We call $y \in X$ an \emph{$\alpha$-approximate median} of $x$ if\tnote{Unfortunately, we can't combine these into a single inequation. Consider the case where $x_i=y$ for all $i$; both LHSs are $1$.} $$\frac{1}{n} \left|\left\{ i \in [n] : x_i \leq y \right\}\right| \geq \frac12-\alpha \qquad \text{and} \qquad \frac{1}{n} \left|\left\{ i \in [n] : x_i \geq y \right\}\right| \geq \frac12-\alpha.$$
\end{defn}
Proposition \ref{prop:thresh_adaptive_lb} is obtained by combining Lemmas \ref{lem:median_reduction} and \ref{lem:thresh_packing} below.

\begin{lem}\label{lem:median_reduction}
Suppose $M$ answers $k = \lceil 1 +  \log_2 T \rceil$ adaptively-chosen queries from $\thresh{[T]}$ with $\eps$-differential privacy and $(\alpha, \beta)$-accuracy. Then there exists an $\eps$-differentially private $M' : [T]^n \to [T]$ that computes an $\alpha$-approximate median with probability at least $1- \beta$.
\end{lem}

\begin{proof}
The algorithm $M'$, formalized in Figure \ref{fig:Mmedian}, uses $M$ to perform a binary search.
\begin{figure}[h!]
\begin{framed}
\begin{algorithmic}
\INDSTATE[0]{Input: $\db \in X^n$.}
\INDSTATE[0]{$M$ is given $\db$.}
\INDSTATE[0]{Initialize $\ell_1=0$, $u_1 = T$, and $j=1$.}
\INDSTATE[0]{While $u_j-\ell_j>1$ repeat:}
\INDSTATE[1]{Let $m_j = \lceil(u_j+\ell_j)/2\rceil$.}
\INDSTATE[1]{Give $M$ the query $c_{m_j} \in \thresh{[T]}$ and obtain the answer $a_j \in [0,1]$.}
\INDSTATE[1]{If $a_j \geq \frac12$, set $(\ell_{j+1},u_{j+1}) = (\ell_j,m_j)$; otherwise set $(\ell_{j+1},u_{j+1}) = (m_j,u_j)$.}
\INDSTATE[1]{Increment $j$.}
\INDSTATE[0]{Output $u_j$.}
\end{algorithmic}
\end{framed}
\vspace{-6mm}
\caption{$M' : X^n \to X$}\label{fig:Mmedian}
\end{figure}

We have $u_1-\ell_1=T$ and, after every query $j$, $u_{j+1}-\ell_{j+1} \leq \lceil(u_j-\ell_j)/2\rceil$. Since the process stops when $u_j-\ell_j=1$, it is easy to verify that $M'$ makes at most $\lceil 1+\log_2(T-1)\rceil$ queries to $M$. 

Suppose all of the answers given by $M$ are $\alpha$-accurate. This happens with probability at least $1-\beta$. We will show that, given this, $M'$ outputs an $\alpha$-approximate median, which completes the proof.

We claim that $c_{u_j}(x) \geq \frac12-\alpha$ for all $j$. This is easily shown by induction. The base case is $c_T(x)=1 \geq \frac12 - \alpha$. At each step either $u_{j+1}=u_j$ (in which case the induction hypothesis can be applied) or $u_{j+1}=m_j$; in the latter case our accuracy assumption gives $$c_{u_{j+1}}(x)=c_{m_j}(x) \geq a_j - \alpha \geq \frac12 - \alpha.$$ 

We also claim that $c_{\ell_j}(x) < \frac12+\alpha$ for all $j$. This follows from a similar induction and completes the proof.
\end{proof}

\begin{lem}\label{lem:thresh_packing}
Let $M : [T]^n \to [T]$ be an $\eps$-differentially private algorithm that computes an $\alpha$-approximate median with confidence $1- \beta$. Then 
\[n \ge \Omega\left( \frac{\log T + \log(1/\beta)}{\alpha\eps}\right).\]
\end{lem}

\begin{proof}
Let $m= \lceil(\frac12-\alpha)n\rceil-1$. 
For each $t \in [T]$, let $x^t \in [T]^n$ denote the dataset containing $m$ copies of $1$, $m$ copies of $T$, and $n-2m$ copies of $t$. Then for each $t \in [T]$, $$\pr{}{M(x^t) = t} \ge 1-\beta.$$
On the other hand, by the pigeonhole principle, there must exist $t_* \in [T-1]$ such that $$\pr{}{M(x^T)=t_*} \leq \frac{\pr{}{M(x^T) \in [T-1]}}{T-1} \leq \frac{\beta}{T-1}.$$ The inputs $x^T$ and $x^{t_*}$ differ in at most $n-2m \leq 2\alpha n + 2$ entries. By group privacy,
$$1-\beta \leq \pr{}{M(x^{t_*}) = t_*} \leq e^{\eps(2\alpha n+2)} \pr{}{M(x^T)=t_*} \leq e^{\eps(2\alpha n+2)} \frac{\beta}{T-1}.$$
Rearranging these inequalities gives $$O(\eps\alpha n) \geq \eps(2\alpha n+2) \geq \log \left( \frac{(1-\beta)(T-1)}{\beta} \right) \geq \Omega(\log(T/\beta)),$$ which yields the result.
\end{proof}

\begin{remark} \label{rem:NonAdaptThresh}
Proposition \ref{prop:thresh_adaptive_lb} can be extended to \emph{online non-adaptive} queries, which yields a separation between the online non-adaptive and offline models for pure differential privacy and threshold queries.
\end{remark}
The key observation behind remark \ref{rem:NonAdaptThresh} is that, while Lemma \ref{lem:median_reduction} in general requires making adaptive queries, for the inputs  $x^t \in [T]^n$ ($t \in [T]$) used in Lemma \ref{lem:thresh_packing} the queries are ``predictable.'' In particular, on input $x^t$, the algorithm $M'$ from the proof of Lemma \ref{lem:median_reduction} will (with probability at least $1-\beta$) always make the same sequence queries. This allows the queries to be specified in advance in a non-adaptive manner. More precisely, we can produce an algorithm $M'_t$ that produces non-adaptive online queries by simulating $M'$ on input $x^t$ and using those queries. Given the answers to these online non-adaptive queries, $M'_t$ can either accept or reject its input depending on whether the answers are consistent with the input $x^t$; $M'_t$ will accept $x^t$ with high probability and reject $x^{t'}$ for $t' \ne t$ with high probabiliy. The proof of Lemma \ref{lem:thresh_packing} can be carried out using $M'_{t_*}$ instead of $M'$ at the end.

\subsection{The $\mathsf{BetweenThresholds}$ Algorithm} \label{sec:BetweenThresholds}
\bigjnote{To keep the terminology more consistent can we explicitly refer to the ``Above Threshold'' algorithm instead of ``sparse vector''?  The name ``sparse vector'' is dumb anyway and should be deprecated :)  But maybe we can say that our algorithm builds on the ``Above Threshold'' algorithm, which underlies the ubiquitous ``sparse vector'' technique?}\tnote{However, readers are more likely to be familiar with the term "sparse vector" than with "above threshold"}

The key technical novelty behind our algorithm for answering adaptively-chosen threshold queries is a refinement of the ``Above Threshold'' algorithm \cite[\S3.6]{DworkR14}, which underlies the ubiquitous ``sparse vector'' technique \cite{DworkNRRV09,RothR10,DworkNPR10,HardtR10}.

The sparse vector technique addresses a setting where we have a stream of $k$ (adaptively-chosen) low-sensitivity queries and a threshold parameter $t$. Instead of answering all $k$ queries accurately, we are interested in answering only the ones that are above the threshold $t$ -- for the remaining queries, we only require a signal that they are below the threshold. Intuitively, one would expect to only pay in privacy for the queries that are actually above the threshold. And indeed, one can get away with sample complexity proportional to the number of queries that are above the threshold, and to the \emph{logarithm} of the total number of queries.

We extend the sparse vector technique to settings where we demand slightly more information about each query beyond whether it is below a single threshold. In particular, we set two thresholds $\lT < \uT$, and for each query, release a signal as to whether the query is below the lower threshold, above the upper threshold, or between the two thresholds. 

As long as the thresholds are sufficiently far apart, whether (the noisy answer to) a query is below the lower threshold or above the upper threshold is \emph{stable}, in that it is extremely unlikely to change on neighboring datasets. As a result, we obtain an $(\eps, \delta)$-differentially private algorithm that achieves the same accuracy guarantees as the traditional sparse vector technique, i.e. sample complexity proportional to $\log k$.

Our algorithm is summarised by the following theorem.\footnote{In Theorem \ref{thm:AboveBelowThreshold}, only one threshold is allowed. However, our algorithm is more general and permits the setting of two thresholds. We have chosen this statement for simplicity.}

\begin{thm} \label{thm:AboveBelowThreshold}
Let $\alpha,\beta,\eps,\delta,t \in (0,1)$ and $n,k \in\N$ satisfy $$n \geq \frac{1}{\alpha \eps} \max \left\{ 12 \log(30/\eps\delta), 16 \log((k+1)/\beta) \right\}.$$
Then there exists a $(\eps,\delta)$-differentially private algorithm that takes as input $x \in X^n$ and answers a sequence of adaptively-chosen queries $q_1, \cdots, q_k : X^n \to [0,1]$ of sensitivity $1/n$ with $a_1, \cdots, a_{\leq k} \in \{\leftsymb, \rightsymb, \haltsymb\}$ such that, with probability at least $1-\beta$,
\begin{itemize}
\item $a_j = \leftsymb \implies q_j(x) \le t$,
\item $a_j = \rightsymb \implies q_j(x) \ge t$, and
\item $a_j = \haltsymb \implies t - \alpha \le q_j(x) \le t + \alpha$.
\end{itemize}
The algorithm may halt before answering all $k$ queries; however, it only halts after outputting $\haltsymb$.
\end{thm}

Our algorithm is given in Figure \ref{alg:between-thresholds}. The analysis is split into Lemmas \ref{lem:bt-privacy} and \ref{lem:bt-accuracy}.

\begin{figure}[h!]
\begin{framed}
\begin{algorithmic}
\INDSTATE[0]{Input: $x \in X^n$.}
\INDSTATE[0]{Parameters: $\eps,\lT,\uT \in (0,1)$ and $n, k \in \N$.}
\INDSTATE[0]{Sample $\mu \sim \Lap(2/\eps n)$ and initialize noisy thresholds $\hlT = \lT + \mu$ and $\huT = \uT - \mu$.}
\INDSTATE[0]{For $j = 1, 2, \cdots, k$:}
\INDSTATE[1]{Receive query $q_j : X^n \to [0,1]$.}
\INDSTATE[1]{Set $c_j = q_j(x) + \nu_j$ where $\nu_j \sim \Lap(6/\eps n)$.}
\INDSTATE[1]{If $c_j < \hlT$, output $\leftsymb$ and continue.}
\INDSTATE[1]{If $c_j > \huT$, output $\rightsymb$ and continue.}
\INDSTATE[1]{If $c_j \in [\hlT,\huT]$, output $\haltsymb$ and halt.}
\end{algorithmic}
\end{framed}
\vspace{-6mm}
\caption{$\mathsf{BetweenThresholds}$} \label{alg:between-thresholds}
\end{figure}

\begin{lem}[Privacy for $\mathsf{BetweenThresholds}$] \label{lem:bt-privacy}
Let $\eps,\delta \in (0,1)$ and $n \in \N$. Then $\mathsf{BetweenThresholds}$ (Figure \ref{alg:between-thresholds}) is $(\eps, \delta)$-differentially private for any adaptively-chosen sequence of queries as long as the gap between the thresholds $\lT, \uT$ satisfies
\[\uT - \lT \ge \frac{12}{\eps n}\left( \log (10/\eps) + \log(1/\delta) + 1\right).\]
\end{lem}

\begin{lem}[Accuracy for $\mathsf{BetweenThresholds}$] \label{lem:bt-accuracy}
Let $\alpha, \beta,\eps,\lT,\uT \in (0,1)$ and $n,k \in \N$ satisfy \[n \geq \frac{8}{\alpha \eps}\left(\log(k+1) + \log(1/\beta)\right).\] Then, for any input $x \in {X}^n$ and any adaptively-chosen sequence of queries $q_1, q_2, \cdots, q_k$, the answers $a_1, a_2, \cdots a_{\leq k}$ produced by $\mathsf{BetweenThresholds}$ (Figure \ref{alg:between-thresholds}) on input $x$ satisfy the following with probability at least $1-\beta$. For any $j \in [k]$ such that $a_j$ is returned before $\mathsf{BetweenThresholds}$ halts,
\begin{itemize}
\item $a_j = \leftsymb \implies q_j(x) \le \lT + \alpha$,
\item $a_j = \rightsymb \implies q_j(x) \ge \uT - \alpha$, and
\item $a_j = \haltsymb \implies \lT - \alpha \le q_j(x) \le \uT + \alpha$.
\end{itemize}
\end{lem}

Combining Lemmas \ref{lem:bt-privacy} and \ref{lem:bt-accuracy} and setting $\lT=t-\alpha/2$ and $\uT=t+\alpha/2$ yields Theorem \ref{thm:AboveBelowThreshold}.

\begin{proof}[Proof of Lemma \ref{lem:bt-privacy}]
Our analysis is an adaptation of Dwork and Roth's \cite[\S3.6]{DworkR14} analysis of the AboveThreshold algorithm. Recall that a transcript of the execution of $\mathsf{BetweenThresholds}$ is given by $a \in \{\leftsymb, \rightsymb, \haltsymb\}^*$. Let $\cM : X^n \to \{\leftsymb, \rightsymb, \haltsymb\}^*$ denote the function that simulates $\mathsf{BetweenThresholds}$ interacting with a given adaptive adversary (cf. Figure \ref{fig:AdaptiveAlg}) and returns the transcript.

Let $S \subset \{\leftsymb, \rightsymb, \haltsymb\}^*$ be a set of transcripts. Our goal is to show that for adjacent datasets $x \sim x'$,
\[\pr{}{\cM(x) \in S} \le e^{\eps}\pr{}{\cM(x') \in S} + \delta.\]

Let $$z^* = \frac{1}{2}(\uT - \lT) - \frac{6}{\eps n} \log(10/\eps) - 1/n \geq \frac{2}{\eps n} \log(1/\delta) .$$ Our strategy will be to show that as long as the noise value $\mu$ is under control, in particular if $\mu \leq z^*$, then the algorithm behaves in essentially the same way as the standard AboveThreshold algorithm. Meanwhile, the event $\mu > z^*$ which corresponds to the (catastrophic) event where the upper and lower thresholds are too close or overlap, happens with probability at most $\delta$.

The following claim reduces the privacy analysis to examining the probability of obtaining any single transcript $a$:

\begin{clm} \label{clm:bt-reduction}
Suppose that for any transcript $a \in \{\leftsymb, \rightsymb, \haltsymb\}^*$, and any $z \le z^*$, that
\[\pr{}{\cM(x) = a | \mu = z} \le e^{\eps/2}\pr{}{\cM(x') = a | \mu = z + 1/n}.\]
Then $\cM$ is $(\eps, \delta)$-differentially private.
\end{clm}

\begin{proof}
By properties of the Laplace distribution, since $\mu \sim \Lap(2/\eps n)$, for any $z \in \R$, we have
\[\pr{}{\mu = z} \le e^{\eps / 2} \pr{}{\mu = z+1/n},\]
and
\[\pr{}{\mu > z^*} = \frac12 e^{-\eps n z^*/2} \le \delta.\]
Fix a set of transcripts $S$. Combining these properties allows us to write
\begin{align*}
\pr{}{\cM(x) \in S} &= \int_{\R} \pr{}{\cM(x) \in S | \mu = z} \pr{}{\mu = z} \mathrm{d}z\\
&\le \left(\int_{-\infty}^{z^*} \pr{}{\cM(x) \in S | \mu = z} \pr{}{\mu = z} \mathrm{d}z\right) + \pr{}{\mu > z^*} \\
&\le \left(e^{\eps / 2}\int_{-\infty}^{z^*} \pr{}{\cM(x') \in S | \mu = z + 1/n} \pr{}{\mu = z} \mathrm{d}z\right) + \delta \\
&\le \left(e^{\eps}\int_{-\infty}^{z^*} \pr{}{\cM(x') \in S | \mu = z + 1/n} \pr{}{\mu = z + 1/n} \mathrm{d}z\right) + \delta \\
&\le e^{\eps} \pr{}{\cM(x') \in S} + \delta
\end{align*}
\end{proof}

Returning to the proof of Lemma \ref{lem:bt-privacy}, fix a transcript $a \in \{\leftsymb, \rightsymb, \haltsymb\}^*$. Our goal is now to show that $\cM$ satisfies the hypotheses of Claim \ref{clm:bt-reduction}, namely that for any $z \le z^*$,
\begin{equation} \pr{}{\cM(x) = a | \mu = z} \le e^{\eps/2}\pr{}{\cM(x') = a | \mu = z + 1/n}.\label{eqn:PointwiseGoal}\end{equation}
For some $k \ge 1$, we can write the transcript $a$ as $(a_1, a_2, \dots, a_k)$, where $a_j \in \{\leftsymb, \rightsymb\}$ for each $j < k$, and $a_k = \haltsymb$. 

For convenience, let $A = \cM(x)$ and $A' = \cM(x')$. We may decompose
\begin{align}
\pr{}{\cM(x) = a | \mu = z} &= \pr{}{ (\forall j < k, A_j = a_j)\ \land\ q_{k}(x) + \nu_{k} \in [\hlT, \huT]| \mu = z} \nonumber \\
&= \pr{}{(\forall j < k, A_j = a_j)| \mu = z} ~ \cdot ~ \pr{}{q_{k}(x) + \nu_{k} \in [\hlT, \huT] | \mu = z \land (\forall j < k, A_j = a_j)}. \label{eqn:BetweenThresholdsDecompose}
\end{align}

We upper bound each factor on the right-hand side separately. 

\begin{clm} \label{clm:FirstFactor}
$$\pr{}{(\forall i < k, A_i = a_i)| \mu = z} \leq  \pr{}{(\forall i < k, A'_i = a_i) | \mu = z + 1/n}$$
\end{clm}
\begin{proof}
For fixed $z$, let $A_z(x)$ denote the set of noise vectors $(\nu_1, \dots, \nu_{k-1})$ for which $(A_1, \dots, A_{k-1}) = (a_1, \dots, a_{k-1})$ when $\nu = z$. We claim that as long as $z \le z^*$, then $A_z(x) \subseteq A_{z+1/n}(x')$. To argue this, let $(\nu_1, \dots, \nu_{k-1}) \in A_z(x)$. Fix an index $j \in \{1, \dots, k-1\}$ and suppose $a_j = \leftsymb$. Then $q_j(x) + \nu_j < \lT + z$, but since $q_j$ has sensitivity $1/n$, we also have $q_j(x') + \nu_j < \lT + (z + 1/n)$. Likewise, if $a_j = \rightsymb$, then $q_j(x) + \nu_j > \uT - z$, so
\[q_j(x') + \nu_j > \uT - z - 1/n \ge \lT + (z + 1/n)\]
as long as $z \le z^* \leq \frac{1}{2}(\uT - \lT) - 1/n$. (This ensures that $\cM(x')$ does not output $\leftsymb$ on the first branch of the ``if'' statement, and proceeds to output $\rightsymb$.)

Since $A_z(x) \subseteq A_{z+1/n}(x')$, this proves that
\begin{align*}
\pr{}{(\forall i < k, A_i = a_i)| \mu = z} &= \pr{}{(\nu_1, \dots, \nu_{k-1}) \in A_z(x)} \\
&\le \pr{}{(\nu_1, \dots, \nu_{k-1}) \in A_{z + 1/n}(x')} \\
&= \pr{}{(\forall i < k, A'_i = a_i) | \mu = z + 1/n}.
\end{align*}
\end{proof}

Given Claim \ref{clm:FirstFactor}, all that is needed to prove \eqref{eqn:PointwiseGoal} and, thereby, prove Lemma \ref{lem:bt-privacy} is to bound the second factor in \eqref{eqn:BetweenThresholdsDecompose} --- that is, we must only show that
\begin{equation} \!\!\!\!\!\!\pr{}{q_{k}(x) + \nu_{k} \in [\hlT, \huT] | \mu = z \land (\forall j < k, A_j = a_j)} \leq e^{\eps/2} \pr{}{q_{k}(x') + \nu_{k} \in [\hlT, \huT] | \mu = z + 1/n \land (\forall j < k, A'_j = a_j)}. \end{equation}

Let $\Delta = (q_{k}(x') - q_{k}(x)) \in [-1/n,  1/n]$. Then
\begin{align*}
&\pr{}{q_{k}(x) + \nu_{k} \in [\hlT, \huT] | \mu = z \land (\forall j < k, A_j = a_j)} \\
&={} \pr{}{\lT + z \le q_{k}(x) + \nu_{k} \le \uT - z} \\
&={} \pr{}{\lT + z + \Delta \le q_{k}(x') + \nu_{k} \le \uT - z + \Delta} \\
&={} \pr{}{\lT + (z +1/n)+ (\Delta-1/n) \le q_{k}(x') + \nu_{k} \le \uT - (z+1/n) + (\Delta+1/n)} \\
&={} \pr{}{q_{k}(x') + \nu_{k} \in [\hlT + \Delta - 1/n, \huT + \Delta +1/n] | \mu = z + 1/n} \\
&\le{} e^{\eps/2}\pr{}{q_{k}(x') + \nu_{k} \in [\hlT, \huT] | \mu = z + 1/n} \\
&={}e^{\eps/2}\pr{}{q_{k}(x') + \nu_{k} \in [\hlT, \huT] | \mu = z + 1/n \land (\forall j < k, A'_j = a_j)} \\
\end{align*}
where the last inequality follows from Claim \ref{clm:Lap} below (setting $\eta = 2/n$, $\lambda=6/\eps n$, $[a,b] = [\hlT, \huT]$, and $[a',b']=[\hlT + \Delta - 1/n, \huT + \Delta +1/n]$) and the fact that $z \le z^* = \frac{1}{2}(\uT - \lT) - \frac{6}{\eps n} \log(10/\eps) - 1/n$ implies
\[b-a=\huT-\hlT = \uT-\lT - 2 \mu \ge \frac{12}{\eps n} \log \left(\frac{10}{\eps} \right) \geq 2\lambda \log \left(\frac{1}{1 - e^{-\eps/6}} \right)\]
whenever $0 \leq \eps \leq 1$.

\begin{clm} \label{clm:Lap}  \mnote{Definitely move to appendix}\tnote{I vote to keep it here.}
Let $\nu \sim Lap(\lambda)$ and let $[a,b], [a',b'] \subset \R$ be intervals satisfying $[a,b] \subset [a',b']$. If $\eta \geq (b'-a')-(b-a)$, then
\[\pr{}{\nu \in [a',b']} \le \frac{e^{\eta/\lambda}}{1 - e^{-(b-a)/2\lambda}} \cdot \pr{}{\nu \in [a,b]}.\]
\end{clm}

\begin{proof}
Recall that the probability density function of the Laplace distribution is given by $f_\lambda(x) = \frac{1}{2 \lambda}e^{-|x|/\lambda}$.
There are four cases to consider: In the first case, $a < b \le 0$. In the second case, $a < 0 < b$ with $|a| \le |b|$. In the third case, $0 \le a < b$. Finally, in the fourth case, $a < 0 < b$ with $|a| \ge |b|$. Since the Laplace distribution is symmetric, it suffices to analyze the first two cases.

\paragraph{Case 1:} Suppose $a < b \le 0$. Then
\begin{align*}
\pr{}{\nu \in [a',b']} 
&\le{} \pr{}{\nu \in [a, b]} + \int_{b}^{b+\eta} \frac{1}{2 \lambda} e^{x / \lambda} \mathrm{d}x \\
&= \frac{1}{2}(e^{(b+\eta)/\lambda} - e^{a/\lambda}) \\
&= \frac{1}{2} \cdot \left(\frac{e^{\eta/\lambda} - e^{(a-b)/\lambda}}{1 - e^{(a - b)/\lambda}} \right) \cdot (e^{b/\lambda} - e^{a/\lambda}) \\
&= \left( \frac{e^{\eta/\lambda} - e^{-(b-a)/\lambda}}{1 - e^{-(b-a)/\lambda}} \right) \cdot\pr{}{\nu \in [a, b]}.
\end{align*}

\paragraph{Case 2:} Suppose $a < 0 < b$ and $|a| \le |b|$. Note that this implies $b \ge (b-a)/2$. Then
\begin{align*}
\pr{}{\nu \in [a',b']} 
&\leq \pr{}{\nu \in [a, b]} + \eta \cdot \frac{1}{2\lambda} e^{a/\lambda}\\
&\leq \pr{}{\nu \in [a, b]}\left( 1 + \frac{\eta}{2\lambda} \frac{e^{a/\lambda}}{\pr{}{\nu \in [0,b]}} \right)\\
&= \pr{}{\nu \in [a, b]}\frac{{1-e^{-b/\lambda}} + \frac{\eta}{\lambda} e^{a/\lambda}}{1-e^{-b/\lambda}}\\
&\leq \pr{}{\nu \in [a, b]}\frac{1 + \eta/\lambda}{1-e^{-b/\lambda}} \\
&\leq \pr{}{\nu \in [a, b]} \frac{e^{ \eta/\lambda}}{1-e^{-(b-a)/2\lambda}} .
\end{align*}

\end{proof}
\end{proof}

\begin{proof}[Proof of Lemma \ref{lem:bt-accuracy}]
We claim that it suffices to show that with probability at least $1-\beta$ we have
\[\forall {1 \le j \le k} \qquad |\nu_j| + |\mu| \le \alpha.\]
To see this, suppose $|\nu_j| + |\mu| \le \alpha$ for every $j$. Then, if $a_j = \leftsymb$, we have
\[c_j = q_j(x) + \nu_j < \hlT = \lT + \mu, \qquad \text{whence} \qquad q_j(x) < \lT + |\mu| + |\nu_j| \le \lT + \alpha.\]
Similarly, if $a_j = \rightsymb$, then
\[c_j = q_j(x) + \nu_j > \huT = \uT - \mu , \qquad \text{whence} \qquad q_j(x) > \uT - (|\mu| + |\nu_j|) \ge \uT - \alpha. \]
Finally, if $a_j = \haltsymb$, then
\[ c_j = q_j(x) + \nu_j \in [\hlT,\huT] = [\lT+\mu,\uT - \mu], \qquad \text{whence} \qquad \lT - \alpha \le q_j(x) \le \uT + \alpha.\]

We now show that indeed $|\nu_j| + |\mu| \le \alpha$ for every $j$ with high probability. By tail bounds for the Laplace distribution,
\[\pr{}{|\mu| > \alpha/4} = \exp\left( -\frac{\eps\alpha n}{8}\right) \qquad \text{and} \qquad \pr{}{|\nu_j| > 3\alpha/4} = \exp\left(-\frac{\eps \alpha n}{8} \right)\]
for all $j$. By a union bound,
\[\pr{}{|\mu| > \alpha/4 ~~~\vee~~~ \exists j \in [k] ~~ |\nu_j| > 3\alpha/4} \le  (k+1) \cdot \exp\left( -\frac{\eps\alpha n}{8}\right) \leq \beta,\]
as required.
\end{proof}

\subsection{The Online Interior Point Problem}

Our algorithm extends a result of \cite{BunNSV15} showing how to reduce the problem of privately releasing thresholds to the much simpler \emph{interior point problem}. By analogy, our algorithm for answering adaptively-chosen thresholds relies on solving multiple instances of an online variant of the interior point problem in parallel. In this section, we present the OIP problem and give an $(\eps, \delta)$-differentially private solution that can handle $k$ adaptively-chosen queries with sample complexity $O(\log k)$. Our OIP algorithm is a direct application of the $\mathsf{BetweenThresholds}$ algorithm from Section \ref{sec:BetweenThresholds}.

\begin{defn}[Online Interior Point Problem]
An algorithm $M$ solves the \emph{Online Interior Point (OIP) Problem} for $k$ queries with confidence $\beta$ if, when given as input any private dataset $x \in [0,1]^n$ and any adaptively-chosen sequence of real numbers $y_1, \cdots, y_k \in [0,1]$, with probability at least $1 - \beta$ it produces a sequence of answers $a_1, \cdots, a_k \in \{\leftsymb, \rightsymb\}$ such that  $$\forall j \in \{ 1, 2, \cdots, k\} \qquad y_j < \min_{i \in [n]} x_i \implies a_j = \leftsymb, \qquad y_j \ge \max_{i \in [n]} x_i \implies a_j = \rightsymb. $$
(If $\min_{i \in [n]} x_i \le y_j < \max_{i \in [n]} x_i$, then $M$ may output either symbol $\leftsymb$ or $\rightsymb$.)
\end{defn}

\begin{figure}[h!]
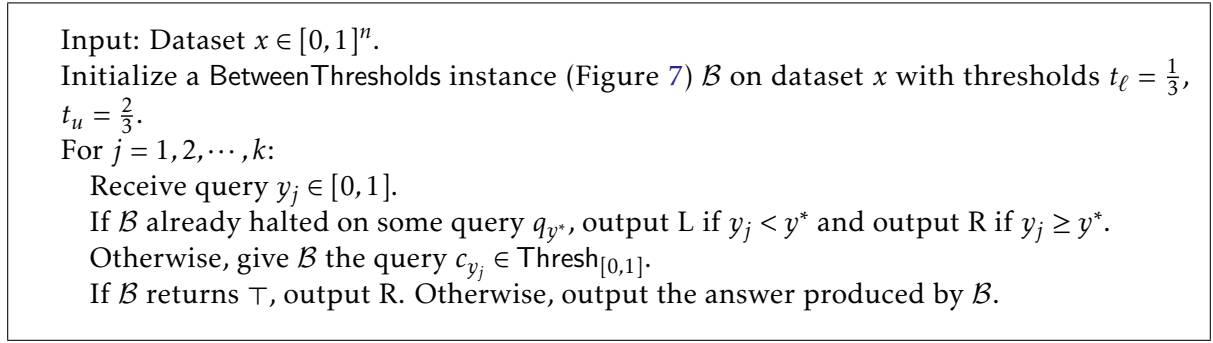

\begin{framed}
\begin{algorithmic}
\INDSTATE[0]{Input: Dataset $x \in [0, 1]^n$.}
\INDSTATE[0]{Initialize a $\mathsf{BetweenThresholds}$ instance (Figure \ref{alg:between-thresholds}) $\mathcal{B}$ on dataset $x$ with thresholds $t_\ell = \frac{1}{3}$, $t_u = \frac{2}{3}$.}
\INDSTATE[0]{For $j= 1, 2, \cdots, k$:}
\INDSTATE[1]{Receive query $y_j \in [0,1]$.}
\INDSTATE[1]{If $\mathcal{B}$ already halted on some query $q_{y^*}$, output $\leftsymb$ if $y_j < y^*$ and output $\rightsymb$ if $y_j \ge y^*$.}
\INDSTATE[1]{Otherwise, give $\mathcal{B}$ the query $c_{y_j}\in \thresh{[0,1]}$.} 
\INDSTATE[1]{If $\mathcal{B}$ returns $\haltsymb$, output $\rightsymb$. Otherwise, output the answer produced by $\mathcal{B}$.}
\end{algorithmic}
\end{framed}
\vspace{-6mm}
\caption{Online Interior Point Algorithm \label{alg:oip}}
\end{figure}

\begin{prop} \label{prop:oip}
The algorithm in Figure \ref{alg:oip} is $(\varepsilon,\delta)$-differentially private and solves the OIP Problem with confidence $\beta$ as long as
\[n \ge \frac{36}{\eps} \left(\log(k+1) + \log(1/\beta) + \log(10/\eps) + \log(1/\delta) + 1\right).\]
\end{prop}

\begin{proof}
Privacy follows immediately from Lemma \ref{lem:bt-privacy}, since Algorithm \ref{alg:oip} is obtained by post-processing Algorithm \ref{alg:between-thresholds}, run using thresholds with a gap of size $1/3$.

To argue utility, let $\alpha = 1/3$ so that
\[n \ge \frac{8}{\eps \alpha}(\log(k+1) + \log(1/\beta)).\]
By Lemma \ref{lem:bt-accuracy}, with probability at least $1-\beta$, the following events occur:
\begin{itemize}
\item If the $\mathsf{BetweenThresholds}$ instance $\mathcal{B}$ halts when it is queried on $c_{y^*}$, then $\min_{i \in [n]} x_i \le y^* < \max_{i \in [n]} x_i$. 
\item If $\mathcal{B}$ has not yet halted and $y_j < \min_{i \in [n]} x_i$, its answer to $c_{y_j}$ is $\leftsymb$.
\item If $\mathcal{B}$ has not yet halted and $y_j \ge \max_{i \in [n]} x_i$, its answer to $c_{y_j}$ is $\rightsymb$.
\end{itemize}
Thus, if $\mathcal{B}$ has not yet halted, the answers provided are accurate answers for the OIP Problem. On the other hand, when $\mathcal{B}$ halts, it has successfully identified an ``interior point'' of the dataset $x$, i.e. a $y^*$ such that $\min_{i \in [n]} x_i \le y^* < \max_{i \in [n]} x_i$. Thus, for any subsequent query $y$, we have that 
\[y < \min_{i \in [n]} x_i \implies y < y^*,\]
so Algorithm \ref{alg:oip} correctly outputs $\leftsymb$. Similarly, 
\[y \ge \max_{i \in [n]} x_i \implies y \ge y^*,\]
so Algorithm \ref{alg:oip} correctly outputs $\rightsymb$ on such a query.
\end{proof}

\subsection{Releasing Adaptive Thresholds with Approximate Differential Privacy}

We are now ready to state our reduction from releasing thresholds to solving the OIP Problem.

\begin{theorem} \label{thm:oip-to-thresholds}
If there exists an $(\eps, \delta)$-differentially private algorithm solving the OIP problem for $k$ queries with confidence $\alpha\beta/8$ and sample complexity $n'$, then there is a $(4\eps, (1 + e^{\eps})\delta)$-differentially private algorithm for releasing $k$ threshold queries with $(\alpha, \beta)$-accuracy and sample complexity
\[n = \max\left\{\frac{6n'}{\alpha}, \frac{24 \log^{2.5}(4/\alpha) \cdot \log(2/\beta)}{\alpha \eps}\right\}.\]
\end{theorem}

Combining this reduction with our algorithm for the OIP Problem (Proposition \ref{prop:oip}) yields:

\begin{corollary} \label{cor:adaptive-thresholds}
There is an $(\eps, \delta)$-differentially private algorithm for releasing $k$ adaptively-chosen threshold queries with $(\alpha, \beta)$-accuracy for
$$n =O\left(\frac{\log k + \log^{2.5}(1/\alpha) + \log(1/\beta\eps\delta)}{\alpha\eps}\right).$$
\end{corollary}

\begin{proof}[Proof of Theorem \ref{thm:oip-to-thresholds}]

\newcommand{\oipalg}{T}\newcommand{\adapthreshalg}[1]{\mathsf{AdaptiveThresholds}_{#1}}
Our algorithm and its analysis follow the reduction of Bun et al.~\cite{BunNSV15} for reducing the (offline) query release problem for thresholds to the offline interior point problem.

Let $\oipalg$ be an $(\eps, \delta)$-differentially private algorithm solving the OIP Problem with confidence $\alpha\beta/8$ and sample complexity $n'$. Without loss of generality, we may assume that $\oipalg$ is differentially private in ``add-or-remove-an-item sense''---i.e. if $x \in [0, 1]^*$ and $x'$ differs from $x$ up to the addition or removal of a row, then for every adversary $\adv$ and set $S$ of outcomes of the interaction between $\adv$ and $\oipalg$, we have $\pr{}{\adaptivealg{\adv}{\oipalg}(x) \in S} \le e^{\eps} \pr{}{\adaptivealg{\adv}{\oipalg}(x') \in S} + \delta$. Moreover, $\oipalg$ provides accurate answers to the OIP Problem with probability at least $1-\alpha\beta/8$ whenever its input is of size at least $n'$. To force an algorithm $\oipalg$ to have these properties, we may pad any dataset of size less than $n'$ with an arbitrary fixed element. On the other hand, we may subsample the first $n'$ elements from any dataset with more than this many elements.

Consider the algorithm $\adapthreshalg{\oipalg}$ in Figures \ref{alg:adaptive-thresholds} and \ref{alg:partition}.

\begin{figure}[h!]
\begin{framed}
\begin{algorithmic}
\INDSTATE[0]{Input: Dataset $x  \in [0, 1]^n$.}
\INDSTATE[0]{Parameter: $\alpha \in (0,1)$.}
\medskip
\INDSTATE[0]{Let $(x^{(1)}, \dots, x^{(M)}) \getsr \mathsf{Partition}(x_1, \dots, x_n, \alpha)$.}
\INDSTATE[0]{Initialize an instance of the OIP algorithm $\oipalg^{(m)}$ on each chunk $x^{(m)} \in [0,1]^*$, for $m \in [M]$.}
\INDSTATE[0]{For each $j = 1, \cdots, k$:}
\INDSTATE[1]{Receive query $c_{y_j} \in \thresh{[0,1]}$.}
\INDSTATE[1]{Give query $y_j \in [0,1]$ to every OIP instance $\oipalg^{(m)}$, receiving answers $a_j^{(1)}, \cdots, a_j^{(M)} \in \{\leftsymb,\rightsymb\}$.}
\INDSTATE[1]{Return $a_j = \frac{1}{M} \cdot \left|\left\{m \in [M] : a_j^{(m)} = \rightsymb\right\}\right|$.}
\end{algorithmic}
\end{framed}
\vspace{-6mm}
\caption{$\adapthreshalg{\oipalg}$ \label{alg:adaptive-thresholds}}
\end{figure}

\begin{figure}[h!]
\begin{framed}
\begin{algorithmic}
\INDSTATE[0]{Input: Dataset $x \in [0, 1]^n$.}
\INDSTATE[0]{Parameter: $\alpha \in (0,1)$.}
\INDSTATE[0]{Output: (Random) partition $(x^{(1)}, \dots, x^{(M)}) \in ([0, 1]^*)^M$ of $x$, where $2/\alpha \le M < 4/\alpha$.}
\medskip
\INDSTATE[0]{Let $M = 2^{\lceil \log_2(2/\alpha) \rceil}$.}
\INDSTATE[0]{Sort $x$ in nondecreasing order $x_1 \le x_2 \le \dots \le x_n$.}
\INDSTATE[0]{For each $0 \le \ell \le \log_2 M$ and $s \in \{0, 1\}^\ell$, sample $\nu_s \sim \Lap((\log_2 M) / \eps)$ independently.}
\INDSTATE[0]{For each $1 \le m \le M-1$, let $\eta_m = \sum_{s \in P(m)} \nu_s$, where $P(m)$ is the set of all prefixes of the binary representation of $m$.}
\INDSTATE[0]{Let $t_0 = 1, t_1 = \left\lfloor \frac{n}{M} + \eta_1 \right\rfloor, \cdots, t_m = \left\lfloor \frac{m \cdot n}{M} + \eta_m \right\rfloor, \cdots, t_{M} = n + 1$.}
\INDSTATE[0]{Let $x^{(m)} = (x_{t_{m - 1}}, \dots, x_{t_m - 1})$ for all $m \in [M]$.}
\end{algorithmic}
\end{framed}
\vspace{-6mm}
\caption{$\mathsf{Partition}$\label{alg:partition}}
\end{figure}

The proof of Theorem \ref{thm:oip-to-thresholds} relies on the following two claims about the $\mathsf{Partition}$ subroutine, both of which are implicit in the work of Bun et al.~\cite[Appendix C]{BunNSV15} and are based on ideas of Dwork et al.~\cite{DworkNPR10}. Claim \ref{clm:partition-privacy} shows that for neighboring databases $x \sim x'$, the behaviors of the $\mathsf{Partition}$ subroutine on $x$ and $x'$ are ``similar'' the following sense: for any fixed partition of $x$, one is roughly as likely (over the randomness of the partition algorithm) to obtain a partition of $x'$ that differs on at most two chunks, where the different chunks themselves differ only up to the addition or removal of a single item. This will allow us to show that running $M$ parallel copies of the OIP algorithm on the chunks remains roughly $(\eps, \delta)$-differentially private. Claim \ref{clm:partition-accuracy} shows that, with high probability, each chunk is simultaneously large enough for the corresponding OIP algorithm to succeed, but also small enough so that treating all of the elements in a chunk as if they were the same element still permits us to get $\alpha$-accurate answers to arbitrary threshold queries.

\begin{clm}\label{clm:partition-privacy}
Fix neighboring datasets $x, x' \in [0, 1]^n$. Then there exists a (measurable) bijection $\varphi: \R^{2M} \to \R^{2M}$ with the following properties:
\begin{enumerate}
\item Let $z \in \R^{2M}$ be any noise vector. Let $x^{(1)}, \dots, x^{(M)}$ denote the partition of $x$ obtained with random noise set to $\nu = z$. Similarly, let $x'^{(1)}, \dots, x'^{(M)}$ denote the partition of $x'$ obtained under noise $\nu = \varphi(z)$. Then there exist indices $i_1, i_2$ such that: 1) For $i \in \{i_1, i_2\}$, the chunks $x^{(i)}$ and $x'^{(i)}$ differ up to the addition or removal of at most one item and 2) For every index $i \notin \{i_1, i_2\}$, we have $x^{(i)} = x'^{(i)}$.
\item For every noise vector $z \in \R^{2M}$, we have $\pr{}{\nu = \varphi(z)} \le e^{2\eps} \pr{}{\nu = z}$.
\end{enumerate}
\end{clm}

\begin{clm}\label{clm:partition-accuracy}
With probability at least $1 - \beta/2$, we have that $|t_m - m \cdot n / M| \leq \alpha n / 24$ for all $m \in [M]$.
\end{clm}

\paragraph{Privacy of Algorithm \ref{alg:adaptive-thresholds}.}
We first show how to use Claim \ref{clm:partition-privacy} to show that Algorithm \ref{alg:adaptive-thresholds} is differentially private. Fix an adversary $A$, and let $B=\adaptivealg{A}{\adapthreshalg{\oipalg}}$ simulate the interaction between $A$ and Algorithm \ref{alg:adaptive-thresholds}. Let $S$ be a subset of the range of $B$. Then, by Property (1) of Claim \ref{clm:partition-privacy} and group privacy\tnote{This is group privacy, not composition.}, we have that for any $z \in \R^{2M}$:
\[\pr{}{B(x)  \in S | \nu = z} \le e^{2\eps} \pr{}{B(x')  \in S | \nu = \varphi(z)} + (1+e^{\eps})\delta.\]
By Property (2) of Claim \ref{clm:partition-privacy}, we also have $\Pr[\nu = z] \le e^{2\eps}\Pr[\nu = \varphi(z)]$ for every $z \in \R^{2M}$. Therefore,
\begin{align*}
\pr{}{B(x)  \in S} &= \int_{ \R^{2M}} \pr{}{B(x)  \in S | \nu = z} \cdot \pr{}{\nu = z} \mathrm{d}z \\
&\le \int_{ \R^{2M}} \left(e^{2\eps}\pr{}{B(x')  \in S | \nu = \varphi(z)} + (1+e^{\eps})\delta\right) \cdot \pr{}{\nu = z} \mathrm{d}z \\
&\le(1+e^{\eps})\delta + \int_{ \R^{2M}} e^{2\eps}\pr{}{B(x')  \in S | \nu = \varphi(z)} \cdot e^{2\eps}\pr{}{\nu = \varphi(z)} \mathrm{d}z \\
&\le(1+e^{\eps})\delta + e^{4\eps} \pr{}{B(x')  \in S }.
\end{align*}
Hence, $B$ is $(e^{4\eps}, (1+e^{\eps})\delta)$-differentially private, as claimed.

\paragraph{Accuracy of Algorithm \ref{alg:adaptive-thresholds}.} We now show how to use Claim \ref{clm:partition-accuracy} to show that Algorithm \ref{alg:adaptive-thresholds} produces $(\alpha, \beta)$-accurate answers. By a union bound, the following three events occur with probability at least $1-\beta$:
\begin{enumerate}
\item For all $m \in [M]$, $\left|\frac{m}{M} - \frac{t_m}{n}\right| \leq \frac{\alpha}{6}$.
\item Every chunk $x^{(m)}$ has size $|x^{(m)}| = t_m-t_{m-1} \in [\alpha n / 6, 2 \alpha n / 3]$.
\item Every instance of $\oipalg$ succeeds.
\end{enumerate}
Now we need to show that if these three events occur, we can produce $\alpha$-accurate answers to every threshold query $c_{y_1}, \dots, c_{y_k}$. Write the sorted input database as $x_1 \le x_2 \le \dots \le x_n$. We consider two cases for the $j^\text{th}$ query: As our first case, suppose $x_n \le y_j$. Then for every chunk $x^{(m)}$, we have $\max \{x^{(m)}\} \le y_j$. Then the success condition of $\oipalg^{(m)}$ guarantees that $a_j^{(m)} = \rightsymb$. Thus, the answer $a_j = 1$ is (exactly) accurate for the query $c_j$.

As our second case, let $i$ be the smallest index for which $x_i > y_j$, and suppose the item $x_i$ is in some chunk $x^{(m_i)}$. Note that this means that 
the true answer to the query $c_{y_j}$ is $(i-1)/n$ and that $t_{m_i-1} \leq i \leq t_{m_i}-1$. Then again, for every $m < m_i$ we have $\max \{x^{(m)}\} \le y_j$, so every such $\oipalg^{(m)}$ instance yields $a_j^{(m)} = \rightsymb$.  Thus,
\[a_j = \frac{1}{M}\cdot \left|\left\{m \in [M] : a_j^{(m)} = \rightsymb\right\}\right| \ge \frac{m_i-1}{M} \ge \frac{t_{m_i}}{n} - \frac{\alpha}{6} -\frac{\alpha}{2} \ge \frac{(i-1)}{n} - \alpha,\]
since $M \geq 2/\alpha$.

On the other hand, for every $m > m_i$, we have $\min \{x^{(m)}\} > y_j$, so every such $\oipalg^{(m)}$ instance instead yields $a_j^{(m)} = \leftsymb$. 
$$a_j \leq \frac{m_i}{M} \leq \frac{t_{m_i}}{n} + \frac{\alpha}{6} \leq \frac{t_{m_i-1} + 2\alpha n / 3}{n} + \frac{\alpha}{6} \leq \frac{i}{n} + \frac{2\alpha}{3} + \frac{\alpha}{6} \leq \frac{i-1}{n} + \alpha,$$ since $n \geq 6/\alpha$.

\end{proof}

\addcontentsline{toc}{section}{Acknowledgements}
\subsection*{Acknowledgements} We thank Salil Vadhan for many helpful discussions.

\addcontentsline{toc}{section}{References}
\bibliographystyle{alpha}
\bibliography{refs}

\appendix

\section{The Fingerprinting Lemma} \label{sec:Fingerprinting}

In this section we prove the fingerprinting lemma (Lemma \ref{lem:Fingerprinting}). The proof is broken into several lemmata. \jnote{Use the proper plural of lemma!  It's so good!}

\begin{lem} \label{lem:Corr}
Let $f : \{\pm 1\}^n \to \mathbb{R}$. Define $g : [\pm 1] \to \mathbb{R}$ by $$g(p) = \ex{x_{1 \cdots n} \sim p}{f(x)}.$$ Then $$\ex{x_{1 \cdots n} \sim p}{f(x) \cdot \sum_{i \in [n]} (x_i - p)} = g'(p) \cdot (1-p^2).$$
\end{lem}
A rescaling of this lemma appears in \cite{SteinkeU15a}. The following proof is taken from \cite{DworkSSUV15}.
\begin{proof}
We begin by establishing several identities.

Since $x^2=1$ for $x \in \{\pm 1\}$, we have the identity $$\frac{\mathrm{d}}{\mathrm{d}p} \frac{1+xp}{2} = \frac{x}{2} = \frac{1+xp}{2} \frac{x-p}{1-p^2} $$ for all $x \in \{\pm 1\}$ and $p \in (-1,1)$. By the product rule, we have $$\frac{\mathrm{d}}{\mathrm{d}p} \prod_{i \in [n]} \frac{1+x_ip}{2} = \sum_{i \in [n]}  \left( \frac{\mathrm{d}}{\mathrm{d}p} \frac{1+x_ip}{2}  \right) \prod_{k \in [n] \setminus \{i\}} \frac{1+x_kp}{2} = \sum_{i \in [n]}  \frac{x_i-p}{1-p^2} \prod_{k \in [n]} \frac{1+x_kp}{2}$$ for all $x \in \{\pm 1\}^n$ and $p \in (-1,1)$.

Sampling $x \sim p$ samples each $x \in \{\pm 1\}$ with probability $\frac{1+xp}{2}$. Thus sampling $x_{1 \cdots n} \sim p$, samples each $x \in \{\pm 1\}^n$ with probability $\prod_{i \in [n]} \frac{1+x_ip}{2}$.

Now we can write $$g(p) = \ex{x_{1 \cdots n} \sim p}{f(x)} = \sum_{x \in \{\pm 1\}^n} f(x) \prod_{i \in [n]} \frac{1+x_ip}{2}.$$ Using the above identities gives 
\begin{align*}
g'(p) =& \sum_{x \in \{\pm 1\}^n} f(x) \frac{\mathrm{d}}{\mathrm{d}p}  \prod_{i \in [n]} \frac{1+x_ip}{2}\\
=& \sum_{x \in \{\pm 1\}^n} f(x) \sum_{i \in [n]}  \frac{x_i-p}{1-p^2} \prod_{k \in [n]} \frac{1+x_kp}{2}\\
=& \ex{x_{1 \cdots n} \sim p}{f(x) \sum_{i \in [n]}  \frac{x_i-p}{1-p^2} }\\
\end{align*}
\end{proof}

\begin{lem} \label{lem:Parts}
Let $g : [\pm 1] \to \mathbb{R}$ be a polynomial. Then $$\ex{p \in [\pm 1]}{g'(p) \cdot (1-p^2)} = 2\ex{p \in [\pm 1]}{g(p) \cdot p}.$$
\end{lem}
\begin{proof}
Let $u(p) = 1-p^2$.
By integration by parts and the fundamental theorem of calculus,
\begin{align*}
\ex{p \in [\pm 1]}{g'(p) \cdot (1-p^2)} =& \frac12 \int_{-1}^1 g'(p) (1-p^2) \mathrm{d}p\\
=& \frac12 \int_{-1}^1 g'(p) u(p) \mathrm{d} p\\
=& \frac12 \int_{-1}^1 \left(\frac{\mathrm{d}}{\mathrm{d}p} g(p)u(p)\right) - g(p) u'(p) \mathrm{d}p\\
=& \frac12 \left( g(1)u(1) - g(-1)u(-1) \right) - \frac12\int_{-1}^1 g(p) (-2p) \mathrm{d}p\\
=& 0 + \int_{-1}^1 g(p) p \mathrm{d}p\\
=& 2\ex{p \in [\pm 1]}{g(p) \cdot p}.
\end{align*}
\end{proof}

\begin{prop} \label{prop:CorrErr2}
Let $f : \{\pm 1\}^n \to \mathbb{R}$. Then $$\ex{p \in [\pm 1], x_{1 \cdots n} \sim p}{f(x) \cdot \sum_{i \in [n]} (x_i - p) + (f(x)-\overline{x})^2} \geq \frac{1}{3}.$$
\end{prop}
\begin{proof}
Define $g : [\pm 1] \to \mathbb{R}$ by $$g(p) = \ex{x_{1 \cdots n} \sim p}{f(x)}.$$
By Lemmas \ref{lem:Corr} and \ref{lem:Parts}, $$\ex{p \in [\pm 1], x_{1 \cdots n} \sim p}{f(x) \cdot \sum_{i \in [n]} (x_i - p)} = \ex{p \in [\pm 1]}{g'(p)(1-p^2)} = \ex{p \in [\pm 1]}{2g(p)p}.$$
Moreover, by Jensen's inequality,
\begin{align*}
\ex{p \in [\pm 1], x_{1 \cdots n} \sim p}{(f(x) - \overline{x})^2} \geq& \ex{p \in [\pm 1]}{\left(\ex{x_{1 \cdots n} \sim p}{f(x) - \overline{x}}\right)^2}\\
=& \ex{p \in [\pm 1]}{(g(p)-p)^2}\\
=& \ex{p \in [\pm 1]}{g(p)^2 -2g(p)p + p^2}\\
=& \ex{p \in [\pm 1]}{g(p)^2} - \ex{p \in [\pm 1], x_{1 \cdots n} \sim p}{f(x) \cdot \sum_{i \in [n]} (x_i - p)} + \frac{1}{3}.
\end{align*}
Rearranging yields the result: $$\ex{p \in [\pm 1], x_{1 \cdots n} \sim p}{f(x) \cdot \sum_{i \in [n]} (x_i - p) + (f(x)-\overline{x})^2} \geq \ex{p \in [\pm 1]}{g(p)^2} + \frac{1}{3} \geq \frac{1}{3}.$$
\end{proof}

We also have an alternative version of Proposition \ref{prop:CorrErr2}:
\begin{prop} \label{prop:CorrErr}
Let $f : \{\pm 1\}^n \to \mathbb{R}$. Then $$\ex{p \in [\pm 1], x_{1 \cdots n} \sim p}{f(x) \cdot \sum_{i \in [n]} (x_i - p) + (f(x)-p)^2} \geq \frac{1}{3}.$$
\end{prop}
\begin{proof}
Define $g : [\pm 1] \to \mathbb{R}$ by $$g(p) = \ex{x_{1 \cdots n} \sim p}{f(x)}.$$
By Lemmas \ref{lem:Corr} and \ref{lem:Parts}, $$\ex{p \in [\pm 1], x_{1 \cdots n} \sim p}{f(x) \cdot \sum_{i \in [n]} (x_i - p)} = \ex{p \in [\pm 1]}{g'(p)(1-p^2)} = \ex{p \in [\pm 1]}{2g(p)p}.$$
Moreover, $$\ex{p \in [\pm 1], x_{1 \cdots n} \sim p}{(f(x) -p)^2} = \ex{p \in [\pm 1], x_{1 \cdots n} \sim p}{f(x)^2 -2g(p) p + p^2} \geq 0 - \ex{p \in [\pm 1]}{2g(p) p} + \frac{1}{3}.$$
The result follows by combining the above equality and inequality.
\end{proof}

Finally we restate and prove Lemma \ref{lem:Fingerprinting}

\begin{lem}[Fingerprinting Lemma] \label{cor:CorrEmpErr}
Let $f : \{\pm 1\}^n \to [\pm 1]$. Then $$\ex{p \in [\pm 1], x_{1 \cdots n} \sim p}{f(x) \cdot \sum_{i \in [n]} (x_i - p) + 2\left|f(x)-\overline{x}\right|} \geq \frac{1}{3}.$$
\end{lem}
\begin{proof}
Since $\left|f(x)-\overline{x}\right| \leq 2$, we have $\left|f(x)-\overline{x}\right|^2 \leq 2 \left|f(x)-\overline{x}\right|$. The result thus follows from Proposition \ref{prop:CorrErr2}.
\end{proof}

\end{document}